\documentclass{article} 

\usepackage{amsmath,amssymb,amsthm}

\usepackage[round]{natbib}
\usepackage{thm-restate}
\usepackage{cleveref}
\usepackage{complexity}

\usepackage{tikz}
\usetikzlibrary{arrows.meta,calc}

\usepackage{algorithm}
\usepackage[noend]{algorithmic}

\usepackage{tabularx} 
\usepackage{booktabs}
\usepackage{subcaption}

\usepackage{fullpage}
\usepackage{authblk}

\newtheorem{theorem}{Theorem}[section]
\newtheorem{corollary}[theorem]{Corollary}
\newtheorem{proposition}[theorem]{Proposition}

\theoremstyle{definition}
\newtheorem{definition}[theorem]{Definition}
\newtheorem{example}[theorem]{Example}

\title{Towards Fair and Efficient Public Transportation:\\ A Bus Stop Model}

\author[1]{Martin Bullinger}
\author[2]{Edith Elkind}
\author[3]{Mohamad Latifian}

\affil[1]{Department of Computer Science, University of Oxford, UK}
\affil[2]{School of Engineering,  Northwestern University, USA}
\affil[3]{School of Informatics, University of Edinburgh, UK\protect\\ \vspace*{0.1cm} martin.bullinger@cs.ox.ac.uk, elkind@cs.ox.ac.uk, mohamad.latifian@ed.ac.uk}

\date{}

\DeclareMathOperator*{\argmin}{arg\,min}

\newcommand{\dyn}{\mathbf{dp}}
\newcommand{\opt}{\mathbf{opt}}

\newcommand{\ins}{\mathcal I} 
\newcommand{\agstops}{\mathcal A(\ins)} 
\newcommand{\betterstops}[2][i]{\mathcal P_{#1}(#2)} 
\newcommand{\exstops}{E} 
\newcommand{\appr}{\beta} 
\newcommand{\costappr}{\gamma} 
\newcommand{\bst}[1][\alpha]{$#1$-\textsc{BSP}} 
\newcommand{\knap}{\textsc{Knapsack}}  
\newcommand{\bg}{c} 
\newcommand{\dynlast}{h} 

\allowdisplaybreaks 

\usepackage{url}

\begin{document}

\maketitle 

\begin{abstract}
    We consider a stylized formal model of public transportation, where a set of agents need to travel along a given road, and there is a bus that runs the length of this road. 
    Each agent has a left terminal and a right terminal between which they wish to travel; they can walk all the way, or walk to/from the nearest stop and use the bus for the rest of their journey. 
    The bus can make a fixed number of stops, and
    the planner needs to select locations for these stops. 
    We study notions of efficiency and fairness for this setting. 
    First, we give a polynomial-time algorithm for computing a solution that minimizes the total travel time; 
    our approach can capture further extensions of the base model, such as more general cost functions or existing infrastructure.
    Second, we develop a polynomial-time algorithm that outputs solutions with provable fairness guarantees (such as a variant of the justified representation axiom or $2$-approximate core). Our simulations indicate that our algorithm almost always outputs fair solutions, even for parameter regimes that do not admit theoretical guarantees. 
\end{abstract}


\section{Introduction}

The use of private vehicles is one of the most significant contributors to pollution.
For instance, it is responsible for 43\% of the greenhouse gas emissions in the European Union \citep{Emissions19a}. 
Therefore, providing well-functioning public transport has repeatedly been identified as a key factor in fighting climate change \citep{WRH03a,Chap07a,KwHa16a}.

The need to model and solve problems related to public transport has been under scrutiny from an operations research perspective; see, e.g., \cite{DeHi07a} for an extensive literature survey.
In the optimization literature, the implementation of public transport infrastructure is commonly seen as a two-stage process consisting of a planning phase and an operational phase.
The planning phase is concerned with the design of the transportation network as well as with determining optimal operation frequencies \citep{LaSa67a,SBP74a}.
In the operational phase, the cost of operating public transport should be minimized, e.g., by optimally assigning vehicles to routes or drivers to buses \citep{DaPa95a,WrRo95a}. In both phases, the primary metric used to evaluate the solution quality is the social welfare, i.e., the total/average travel time.

While optimizing the social welfare is a natural and intuitively appealing goal, we believe that it is equally important to approach the design of transportation networks from a fairness perspective. That is, the proposed route networks, frequencies and types of vehicles should benefit not just the majority of the population, but also smaller and less powerful groups, providing usable connections between all neighborhoods and serving the needs of all residents. 

We propose to tackle this challenge using the conceptual apparatus
of group fairness, building on the ideas of justified representation in multi-winner voting~\cite{aziz2017justified} and core stability 
in cooperative game theory~\cite{gillies1959solutions}.
 The intuition that we aim to capture is that sufficiently large groups of agents with similar preferences deserve to be represented in the selected solution, or, more ambitiously, that each group should be allocated resources in proportion to its size.

 While we believe that this perspective should be taken into account at all stages of transportation planning, we showcase our approach by applying it to a specific and relatively simple task: choosing
 the locations of the stops for a fixed bus/train route. 
 Specifically, we consider the setting where the trajectory of the vehicle has been exogenously determined, either by topography (e.g., a mountain road or a river) or by existing infrastructure (e.g., train tracks), the number of stops has been fixed in advance due to bounds on the overall travel time, but the designer still has the freedom to decide where to place the stops. Then, to use the public transport option, the user would have to travel to a nearby stop by using private transport (such as walking, cycling, or using an e-scooter), ride the vehicle towards their destination, and then use private transport again for the last-mile travel. Alternatively, they can opt to use private transport for the entire trip; however, we assume that private transport has higher per-mile cost (measured as physical effort, travel time, or monetary cost) than public transport. Crucially, the agents' decision whether to use the public transport at all is influenced by the location of the stops, so the planner's choices made at this stage may have a dramatic effect on the demand for public transport: positioning the stops without taking into account the agents' travel needs may render the system unusable and push the residents towards private transportation solutions.

For readability, when describing the model, we talk about a bus and the agents walking to/from bus stops; however, we emphasise that our model is applicable to inter-urban transportation, such as train routes and long-distance buses (in which case the agents' last-mile transportation solutions may involve cycling or riding a scooter rather than walking).

\subsection{Our Contribution}\label{sec:contribution}
We put forward a stylized model where there is a bus route that travels the length of a given road, and there are $n$ agents who may ride this bus. Each agent wants to travel between two terminal points located along this road; they can walk all the way, or take the bus (in which case they still need to walk to/from suitable stops).
The planner has a budget to build a limited number of bus stops and is given a set of possible
stop locations; they then decide which stops should be built.
A solution, i.e., a set of bus stop locations, is evaluated according to two criteria.
First, we measure it in terms of efficiency, defined by the total time the agents spend on traveling between their terminals.
Second, we investigate to what extent a solution offers proportional representation to agent groups.
We assume that each of the $n$ agents is entitled to the $1/n$ fraction of the available budget. We then want to achieve outcomes that are group-fair, in the sense that there is no set of agents $S$ such that all agents 
in $S$ can withdraw their shares of the budget and then pool them to build a pair (resp., a set)
of stops such that
all agents in $S$ prefer the outcome where only these stops are built to the current outcome;
we say that solutions with this property 
provide justified representation (resp., lie in the core).

Our first contribution is a dynamic program that can efficiently compute cost-minimal solutions.
This approach is very flexible in that it still works when we add further features to the model, such as travel costs dependent on non-homogeneous road conditions or existing infrastructure.
Moreover, while computing the minimum total cost becomes {\NP}-complete when bus stops have variable costs, 
our dynamic program still runs in pseudo-polynomial time with respect to the budget.

In the second part of the paper, we focus on finding solutions that provide justified representation (JR) or are (approximately) in the core.
Unfortunately, efficiency and justified representation turn out to be incompatible.
However, we present a polynomial-time algorithm that operates by selecting bus stops at 
distances proportional to the density of terminal points, and show that this algorithm finds JR solutions whenever the cost of taking the bus is zero.
Moreover, this algorithm offers a $2$-approximation to the stronger fairness concept of the core, and exhibits excellent empirical performance (on synthetic data).
In contrast, there are instances for which no solution can provide a stronger form of JR.

\subsection{Related Work}
Fairness considerations have a long-standing history in collective decision-making \citep[see, e.g.,][]{Rawl71a,Sen09a}.
In the context of transportation, fairness is often concerned with justice in terms of equity. 
It is then measured in terms of, e.g., availability to monetarily disadvantaged population \citep{Puch82a}, distribution of the impact on health caused by pollution \citep{FoSc99a}, or general access to key infrastructure \citep{PSB17a}.

Fairness in transportation has been studied in the operations research literature, but the existing work is limited to the operational phase of transportation.
For instance, \citet{JST09a} aim at fairly levelling road occupation to avoid congestion, while \citet{MHV18a} are concerned with balancing the workload among a fleet of vehicles that have to jointly cover a given set of trips.

In contrast, our approach, i.e., modeling fairness in terms of proportionality, is rooted in the (computational) social choice literature \citep[see e.g.,][]{ConitzerF017,JiangMW20,PetersS20}.
Our model can be viewed as a special case of multi-winner voting, 
and our notion of justified representation is 
an adaptation of a similar concept in multi-winner approval voting~\cite{aziz2017justified}.
It is also
similar to the notion of proportional fairness in fair clustering \citep{CFLM19a,MiSh20a}.
In this stream of literature, \citet{LLS+21a} study approximate core stability, 
where the approximation is with respect to the size of the deviating coalition 
(which is similar in spirit to our approach) or with respect to the gain by the deviating agents.
\citet{KKK23a} consider a similar approximation in the context of multi-winner voting, 
and \citet{ChaudhuryLK0M22} explore similar ideas in federated learning settings.

We note that placing stops on the line is similar in spirit to facility location~\cite{ChanFLLW21};
however, our focus in this work is on fairness, whereas much of the facility location literature takes a mechanism design perspective (see, however, \cite{ZhouLC22,ElkindLZ22}).
Most related to our paper are models which investigate the same cost function \cite{FSN20a,ChWa23a}.
In particular, the model by \citet{ChWa23a} is a special case of our model where $\alpha = 0$ (i.e., taking the bus has no cost), all agents have the same destination, and only two bus stops are built.
However, our work differs in two key aspects: we allow for more than two stops to be built and study fairness aspects (rather than strategic manipulation).
The facility location literature also considers agents that are interested in more than one location \cite{SerafinoV14,AnastasiadisD18}, but these works use different cost functions. 

A recent preprint by \citet{HBL+24} also considers fairness in the design of transportation networks, but differs from our work in two aspects. First, the authors model fairness via a welfarist approach, i.e., they consider a family of welfare measures that interpolate between egalitarian and utilitarian welfare. Second, in their model the input is captured by an undirected graph, and the planner's task is to build a subset of edges of this graph. Thus, while the two papers share similar high-level motivation, their technical contributions do not overlap.

\section{Model}
Given a positive integer $k\in \mathbb N$, we write $[k] := \{1,\dots, k\}$.
For two numbers $x,y\in \mathbb Q$, we denote by $d(x,y) := |x-y|$ the Euclidean distance from $x$ to $y$. We extend this notation to sets: given a number $x\in \mathbb Q$ and a set of numbers $P\subseteq \mathbb Q$, we write $d(x,P) :=\min_{y\in P}d(x,y)$. 

For $\alpha \in [0,1]$, 
an instance $\ins = \langle N,V,b, (\theta_i)_{i\in N}\rangle$ of the $\alpha$-\emph{bus stop problem} (\bst) is given by a finite set $N$ of $n$ agents, a finite set $V\subseteq \mathbb Q$ of $m$ potential bus stops, a budget $b\in \mathbb N$, and, for each agent $i\in N$, their \emph{type} $\theta_i = (\ell_{i},r_{i})$, where $\ell_{i}, r_{i}\in V$ and $\ell_i < r_i$. We refer to the points $\ell_i$ and $r_i$
as the {\em terminal points} of $i$.
We denote the set of all terminal points of instance $\ins$ by
$\agstops:=\{\ell_i,r_i\colon i\in N\}$.


A \emph{solution} to an instance 
of \bst{} is a set of bus stops $S\subseteq V$.
A solution $S$ is said to be \emph{feasible} if $|S|\le b$, i.e., the number of selected bus stops does not exceed the budget. 

Our cost function extends models of facility locations in which two stops are build \citep{FSN20a,ChWa23a}.
For each agent $i\in N$, their cost of traveling between
two points $\ell, r\in\mathbb Q$ is $d(\ell,r)$ if they walk and $\alpha\cdot d(\ell,r)$ if they take the bus. Consequently, 
the cost of agent $i$ for a solution $S$ to an instance~$\ins$ 
is given by
\begin{equation}
    c_i^{\ins}(S):= \min 
\begin{cases}
d(\ell_i,r_i) & \text{``agent walks''}\\
\min_{x,y\in S} \left[d(\ell_i,x) + \alpha\cdot d(x,y) + d(r_i,y)\right] & \text{``agent takes bus''}.
\end{cases}
\end{equation}
This expression considers two possibilities for $i$: (1) walking
all the way from $\ell_i$ to $r_i$, or (2) walking from $\ell_i$ to a bus stop $x$, taking the bus to another stop $y$, and then walking from $y$ to $r_i$, where $x$ and~$y$ are chosen to minimize the overall travel cost.

\begin{figure*}
    \centering
    \resizebox{1\textwidth}{!}{
    \begin{tikzpicture}
        \draw[->] (-.7,0) -- (15.7,0);
        \foreach \i in {0,1,2,3,4,5,6,7,8,9,10,11,12,13,14,15}
        {
        \draw (\i,-.1) -- (\i,.1);
        }
        \foreach \i in {0,1,4,7,10,13,15}
        {
        \node at (\i,-.4) {\i};
        }
        \node[align = center] at (-1,-.4) {terminals};
        \node[align = center] at (-1,.9) {desired\\routes};
        \draw (1,.3) edge[->] node[pos =.5, fill = white] {$a_3$} (4,.3);
        \draw (1,.6) edge[->] node[pos =.5, fill = white] {$a_4$} (7,.6);
        \draw (1,.9) edge[->] node[pos =.5, fill = white] {$a_5$} (10,.9);
        \draw (1,1.2) edge[->] node[pos =.5, fill = white] {$a_6$} (13,1.2);
        \draw (0,1.5) edge[->] node[pos =.53, fill = white] {$a_1$, $a_2$} (15,1.5);
    \end{tikzpicture}
    }
    \caption{Illustration of \Cref{ex:incompatible}. 
    The same instance proves the incompatibility of efficiency and JR in \Cref{thm:effVSfair}.}
    \label{fig:effVSfair}
\end{figure*}

The \emph{total cost} of a solution $S\subseteq V$ 
for an instance $\ins$ is defined as $c^{\ins}(S) := \sum_{i\in N}c_i^{\ins}(S)$.
Whenever the instance $\ins$ is clear from the context, we omit the superscript $\ins$.
A solution is {\em efficient} if it minimizes the total cost among feasible solutions.

Apart from efficiency, we are also interested in fairness.
Our first concept of fairness builds on ideas from the multi-winner voting and fair clustering literature \citep{aziz2017justified, CFLM19a,MiSh20a,LLS+21a,aziz24}.
Suppose we are given an instance with $n$ agents and budget $b$. Then, intuitively, each agent is entitled to $\frac bn$ units of money, so a group of $\lceil\frac nb\rceil$ agents should be able to dictate the position of one stop.
Therefore, one may want to rule out solutions $S$ such 
that all agents in a group 
of size at least $\frac nb$ can lower their costs by 
abandoning $S$ and building a single stop.
However, this condition is too weak, as no agent benefits from a single stop.
Hence, we strengthen it by considering groups of agents that are entitled to \emph{two} stops.

\begin{definition}\label{def:PF}
    A solution $S\subseteq V$ is said to provide \emph{justified representation (JR)} if for every set of agents $M\subseteq N$ with $|M|\ge \frac {2n}b$ and every pair of stops $T\subseteq V$ there exists an agent $i\in M$ such that $c_i(T) \ge c_i(S)$.
    Moreover, a solution $S\subseteq V$ is said to provide \emph{strong justified representation} if for every set of agents $M\subseteq N$ with $|M|\ge \frac {2n}b$ and every pair of stops $T\subseteq V$ there exists an agent $i\in M$ such that $c_i(T) > c_i(S)$ or for all agents $i\in M$ it holds that $c_i(T) \ge c_i(S)$.
\end{definition}

The key distinction between JR and strong JR is that, to define the former, we only consider deviations to pairs of stops that are strictly preferred by each agent in $M$, whereas to define the latter, we also consider deviations that make no agent in $M$ worse off while making at least one member of $M$ strictly better off.
Thus, an outcome that provides strong JR also provides JR, but the converse is not necessarily true.

Justified representation can also be viewed as a notion of stability: a group of at least $\frac {2n}b$ agents can 
deviate by building two stops, and we require that there is no group such that all group members can benefit from a deviation.
Note that a budget of $2$ is exactly the proportion of the budget that a group of size $\left\lceil \frac {2n}b \right\rceil$ is entitled to spend.
By generalizing this idea to groups of arbitrary size, where a deviating group is allowed to spend a fraction of the budget that is proportional to the group size, we arrive to the concept of the core. We note that the core has been considered as a notion of fairness in a variety of contexts, ranging from participatory budgeting to clustering \citep{fain2016core,AKMT22a,AMS23a,ChaudhuryLK0M22,ChaudhuryMY0MP24}.

\begin{definition}\label{def:core}
    A subset of agents $M\subseteq N$ is said to \emph{block} a solution $S\subseteq V$ if there exists a subset of stops $T\subseteq V$ such that 
    $|T| \le |M|\cdot \frac bn$
    and $c_i(T) < c_i(S)$ for all agents $i\in M$.
    Moreover, a subset of agents $M\subseteq N$ is said to \emph{weakly block} a solution $S\subseteq V$ if there exists a subset of stops $T\subseteq V$ such that 
    $|T| \le |M|\cdot \frac bn$,
    $c_i(T) \le c_i(S)$ for all agents $i\in M$, and there exists $j\in M$ with $c_j(T) < c_j(S)$.
    A solution is said to be in the \emph{strong core} if it is not weakly blocked.
\end{definition}

Equivalently, a solution $S\subseteq V$ is in the \emph{core} if, for every set of agents $M\subseteq N$ and every set of stops $T\subseteq V$ with $|T| \le |M|\cdot \frac bn$, there exists an agent $i\in M$ such that $c_i(T) \ge c_i(S)$.
The core is a demanding solution concept.
Therefore, we also define a multiplicative approximation of the core (which we call the $\appr$-core), where, for a group of agents to be allowed to deviate by building $t$ stops, the size of the group should be at least $\appr$ times the 
number of agents who `deserve' $t$ stops.
Note that the $1$-core is identical to the core.

\begin{definition}
    Let $\appr \ge 1$.  
    A solution $S\subseteq V$ is said to be in the \emph{$\appr$-core} if, for every set of agents $M\subseteq N$ and every set of stops $T\subseteq V$ with $\appr\cdot|T| \le |M|\cdot \frac bn$, there exists an agent $i\in M$ such that $c_i(T) \ge c_i(S)$.
\end{definition}

We provide an example to illustrate our model.

\begin{example}\label{ex:incompatible}
    Let $\alpha\in [0,1)$.
    Consider the instance $\langle N,V,b, (\theta_i)_{i\in N}\rangle$, depicted in \Cref{fig:effVSfair}, 
    with $N = \{a_i\colon i\in [6]\}$, $V = \{0,1,4,7,10,13,15\}$, and $b = 6=|N|$.
    The agents' types are as follows: $\theta_{a_1} = \theta_{a_2} = (0,15)$, $\theta_{a_3} = (1,4)$, $\theta_{a_4} = (1,7)$, $\theta_{a_5} = (1,10)$, and $\theta_{a_6} = (1,13)$.

    Consider the solution $S^* = \{1,4,7,10,13,15\}$.
    It holds that $c(S^*) = 60\alpha + 2(1-\alpha)$.
    However, $S^*$ does not provide JR.
    To see this, consider $M = \{a_1,a_2\}$ and $T = \{0,15\}$. 
    Then, $|M| = \frac {2n}b$ and $c_{a}(T) < c_{a}(S^*)$ for each $a\in M$.

    In contrast, any solution $S' = V \setminus \{x\}$ for $x\in \{4,7,10,13\}$ provides JR because then $S'$ contains the terminals of all except possibly one agent, who is only entitled to one stop.
    \hfill$\lhd$
\end{example}

We can use Example~\ref{ex:incompatible} to prove that providing JR is incompatible with minimizing total cost, apart from the trivial case of $\alpha = 1$ where walking and taking the bus takes the same time.

\begin{restatable}{proposition}{effVSfair}\label{thm:effVSfair}
    For each $\alpha\in [0,1)$
    there exists an instance of {\bst} such that no feasible solution can both minimize the total cost and provide JR.
\end{restatable}
\begin{proof}
    Consider the instance in Example~\ref{ex:incompatible} and the solution
    $S^* = \{1,4,7,10,13,15\}$.
    We already know that $S^*$ does not provide JR.
    To complete the proof, we show that $S^*$ is the unique solution of minimum cost.
    Recall that $c(S^*) = 60\alpha + 2(1-\alpha)$.
    
    First, note that the sum of lengths of the agents' routes is $2\cdot 15 + 3 + 6 + 9 + 12 = 60$. 
    Hence, the cost of every solution is at least $60\alpha$.
    
    We now show that every other solution costs more than $S^*$.
    Fix a solution $S\subseteq V$ with $|S|=6$.
    If $1 \notin S$, then the walking cost of each of the agents in $\{a_3, a_4, a_5, a_6\}$ is at least $1$, so 
    $c(S) \ge 60\alpha + 4(1-\alpha)$.
    Therefore, we may assume that $1\in S$.
    Moreover, if $15\notin S$, then the walking costs of $a_1$ and $a_2$ are at least~$2$, so the solution is worse than $S^*$.
    Hence, we may also assume that $15\in S$.
    
    Next, assume that $0\in S$ and hence
    $\{0,1,15\}\subseteq S$.
    Therefore, we only have~$3$ stops to cover the right terminals of agents $\{a_3, a_4, a_5, a_6\}$.
    One of these agents has to walk a distance of at least $3$, unless $S = \{0,1,4,7,10,15\}$.
    In the latter case, $a_6$ has to stay on the bus for 2 units of distance past their right terminal and then walk back.
    Hence, $0\in S$ implies $c(S) \ge \min\{60\alpha + 3(1-\alpha), 62\alpha + 2(1-\alpha)\} > c(S^*)$.
    Thus, we conclude that $0\notin S$ and therefore $S=V\setminus\{0\}=S^*$.
    %
\end{proof}

In fact, the incompatibility observed in Proposition~\ref{thm:effVSfair} can be strengthened further: it holds for approximate JR and, in case of the \bst[0], even for approximate minimum cost.
We provide the details for these results in \Cref{app:further}.
However, if we replace cost minimality with Pareto optimality, the incompatibility no longer holds: by applying Pareto improvements, we can transform a solution providing JR into a Pareto-optimal solution providing JR.

\begin{definition}
    Given an instance $\ins = \langle N,V,b, (\theta_i)_{i\in N}\rangle$, 
    a solution $S$ is said to \emph{Pareto-dominate} another solution $S'$ if $c_i(S)\le c_i(S')$ for all $i\in N$ and there exists an agent $j\in N$ with $c_j(S) < c_j(S')$. 
    A solution $S^*$ is {\em Pareto-optimal} for $\ins$ if it is not dominated by any other solution.
\end{definition}

\begin{proposition}\label{prop:pareto:p}
    Let $\alpha \in [0,1]$.
    Then every instance of \bst{} that admits a solution providing JR also admits a solution that is Pareto-optimal and provides JR.
\end{proposition}

\begin{proof}
    Let $\alpha \in [0,1]$, and consider an instance of \bst{} that admits a solution $S$ providing JR. Suppose that $S$ is not Pareto-optimal. Then it is Pareto-dominated by another solution $S_1$. We claim that $S_1$, too, provides JR. To see this, 
    consider a subset of agents $M\subseteq N$ with $|M| \ge \frac{2n}b$ and a pair of bus stops $T\subseteq V$.
    Since $S$ provides JR, there exists an agent $i\in M$ with $c_i(S)\le c_i(T)$.
    Hence, $c_i(S_1)\le c_i(S)\le c_i(T)$, which establishes that $S_1$ provides JR. If $S_1$ is not Pareto-optimal, there exists another solution $S_2$ that Pareto-dominates it, and our argument shows that $S_2$ provides JR as well. We can continue in this manner until we reach a Pareto-optimal solution; this will happen after a finite number of steps, as each step reduces the total cost.
\end{proof}

\Cref{prop:pareto:p} extends to approximate JR solutions; the proof remains the same. 
We note, however, that 
\Cref{prop:pareto:p} does not offer an efficient algorithm to find a Pareto-optimal solution that provides JR, as it is not clear how to compute Pareto improvements in polynomial time.


\section{Efficiency}\label{sec:eff}
In this section, we show that efficient solutions, i.e., solutions of minimum total cost, can be computed in polynomial time. 
Our algorithm 
extends to a more general version of our model, where agents' terminals need not be contained in $V$.

Our algorithm is based on a dynamic program, which iteratively considers adding new stops to the solution.
Capturing our problem by a dynamic program is challenging, because each agent's cost depends on the placement of \emph{two} stops.
The crucial observation that enables us to circumvent this difficulty is that, to perform cost \emph{updates} in the dynamic program, it suffices to know the rightmost stop in the current solution.
The relevant computation is captured in \Cref{lem:update}, which addresses the change in cost resulting from the addition of an extra bus stop.
All missing and full proofs from this section can be found in \Cref{app:missing:eff}.

\begin{restatable}{lemma}{updateLemma}\label{lem:update}
    Let $\alpha \in [0,1]$ and let $\ins = \langle N,V,b, (\theta_i)_{i\in N}\rangle$ be an instance of the \bst{} problem.
    Let $S\subseteq V$, and $h = \max S$. Then for each $k\in V$ with $k > h$
    the quantity $c(S) - c(S\cup \{k\})$ is a function of $h$ and $k$ that can be computed in time $O(1)$.
\end{restatable}

\begin{proof}[Proof Sketch]
    The key observation is that, after $k$ is added, the only update to their route that agent $i$ needs to consider is whether to ride the bus between $h$ and $k$ and then walk to their destination; alternatively, they can keep their original route. Their gain from riding the bus between $h$ and $k$ only depends on $h, k, \ell_i$, and $r_i$. 
\end{proof}

\Cref{lem:update} enables us to set up a two-dimensional dynamic program for computing the minimum total cost of a solution to an {\bst} instance with a given number of stops. 

\begin{restatable}{theorem}{DPoptimal}
\label{thm:optimal}
    For $\alpha \in [0,1]$, we can compute a minimum-cost solution for an instance of {\bst} in time $O(nm^2)$.
\end{restatable}

\begin{proof}
    First, in the case where $b\ge 2n$, the budget is allows us to assign two dedicated bus stops to each agent.
    Then, for each agent we can simply establish the two best bus stops.
    For this, we can first find the closest bus stops to the left and to the right of their two terminals in time $O(m)$.
    If there is a bus stop at the agents position, we only use this bus stop.
    Then, we can check the at most four possible combinations of using the bus from each of the bus stops closest to the agent and compare it to the cost of walking.
    Performing this for each agent leads to a total running time of $O(nm)$.
    The union of the best bust stops yields the smallest cost in any feasible solution for each agent.
    Hence, it suffices to consider instances where $b<2n$, for which we will define our dynamic program.
   
    Let $\alpha \in [0,1]$, and consider an instance $\ins = \langle N,V,b, (\theta_i)_{i\in N}\rangle$ of {\bst}
    where $V = \{v_1,\dots, v_m\}$
    with $v_1\le v_2\le \dots \le v_m$, i.e., potential stops are sorted from left to right and $v_j$ represents the $j$th stop. We will assume $b\le m$, as otherwise there is an optimal-cost solution that builds a stop at each location.
    For a solution $S$ with  $\max S = v_h$ and $k > h$, let $\Delta(h, k) := c(S) - c(S\cup \{v_k\})$ denote the total reduction in the agents' costs from adding stop $v_k$ to~$S$. 
    By \Cref{lem:update}, we know that $\Delta(h, k)$ only depends on $v_k$ and $v_h$, and can be computed in $O(n)$ time. 
    Hence, we can precompute $\Delta(h,k)$ for all values of $h$ and $k$ in time $O(nm^2)$ since both $h$ and $k$ can take at most $m$ values.
    Let us set up a dynamic program $\dyn[\dynlast, \bg]$, where 
    \begin{itemize}
        \item $\dynlast\in \{0,1,\dots, m\}$ represents the rightmost stop that has been added to the solution, where $0$ means that no stop has been added yet, and
        \item $\bg\in \{0,\dots,b\}$ represents the budget used so far.
    \end{itemize}

    Then, $\dyn[\dynlast, \bg]$ is the minimum total cost of a solution that selects at most $\bg$ stops, with  $v_\dynlast$ being the rightmost selected stop.

    We initialize with 
    
    \begin{itemize}
        \item[(a)] $\dyn[0, \bg] = \dyn[1, \bg] = \sum_{i\in N} (r_i - \ell_i)$ for all $\bg\in \{0,\dots, b\}$,
        \item[(b)] $\dyn[\dynlast, 0] = \infty$ for all $\dynlast\in \{1,\dots, m\}$.
    \end{itemize}
    Case~(b) captures the impossible situation of selecting at least one stop ($h>0$) while spending no budget.
    We prevent this case by setting the total cost to $\infty$. 
    As we assume $b\le m$, in total the initialization takes $O(nm)$ time.

     For updating, we use the cost change function $\Delta$. For $\dynlast\in \{1,\dots, m\}$, and $\bg\in \{1,\dots, b\}$ we update as follows:
    \begin{equation}\label{eq:DPupdate}
    \dyn[\dynlast, \bg] 
    = \min_{\dynlast'\in \{0,\dots,\dynlast-1\}}\dyn[\dynlast', \bg-1]
    - \Delta(\dynlast',\dynlast).
    \end{equation}

    That is, we consider the position of the stop $h'$ that precedes $\dynlast$ in the solution, and evaluate the cost reduction from adding $\dynlast$ to a solution that ends with $h'$.
    Clearly, the updates can be computed in time $O(m)$, using the precomputed values $\Delta(h,k)$. 
    As we assume that $b < 2n$, our table has at most $2nm$ entries, and can be filled in time $O(nm^2)$.
    Hence, together with the precomputation of the values $\Delta(h,k)$ in time $O(nm^2)$ and initializing the table in time $O(nm)$, the total running time is $O(nm^2)$.

    It is not hard to see that our dynamic program is correct; we provide 
    a formal proof of correctness
    in \Cref{app:missing:eff}.
    Hence, the minimum cost of a feasible solution for $\ins$ 
    is equal to 
    $\min_{\dynlast\in \{0,\dots,m\}}\dyn[\dynlast, b]\text.$
   
    We can efficiently extract an explicit feasible solution of minimum cost by standard techniques.
\end{proof}

Notably, the computations in the dynamic program developed in the proof of \Cref{thm:optimal} are merely updates of the sums of costs for all agents.
This allows us to extend \Cref{thm:optimal} to incorporate further features that may be important for some applications.

First, the theorem extends to more general cost functions. Consider a setting where the time to travel between two stops depends on factors other than the distance.
For instance, the bus route might encompass intervals with different speed limits, or there may be hilly or curvy roads, where the bus needs to slow down.
Hence, the travel costs need not be homogeneous.
However, typically the travel cost of a route 
only depends on the costs of its segments.
To capture this, we introduce the notion of separable travel costs, and formally introduce the separable-cost BSP problem, which generalizes the \bst{} problem defined earlier in the paper.

A function $d\colon \mathbb Q^2\to \mathbb Q$ is called \emph{separable} if for all $x < y < z$ it holds that $d(x,z) = d(x,y) + d(y,z)$.
An instance $\ins$ of {\em separable-cost BSP} is given by a tuple $\langle N,V,b, (\theta_i)_{i\in N}\rangle$ and separable cost functions $d^W: V\times V\to{\mathbb Q}$
and $d_i^B:V\times V\to{\mathbb Q}$ for all $i\in N$
(where the superscripts refer to walking and taking the bus).
The agents' costs are then defined as
\begin{equation*}
c_i^{\ins}(S):= \min 
\begin{cases}
d^W(\ell_i,r_i) \\ 
\min_{x,y\in S} \left[d^W(\ell_i,x) + d^B(x,y) + d^W(y, r_i)\right]\text. 
\end{cases}
\end{equation*}

We can generalize \Cref{lem:update} to separable-cost BSP by replacing the costs for walking and taking the bus by the respective separable cost functions in all update formulae.
Note that the definition of separable-cost BSP does not require that taking the bus is faster than walking.
Our computation can account for this, by allowing the agents to walk rather than take the bus for segments where walking is faster.
With the generalized update formulae, we can run the dynamic program from \Cref{thm:optimal} and obtain the following theorem.

\begin{theorem}
\label{thm:separable}
    For separable-cost BSP, we can compute the minimum cost in time $O(nm^2)$.
\end{theorem}

As a second extension, we assume that there already is an existing set of bus stops, but we have a budget to build~$b$ additional stops.
An instance of {\bst} \emph{with existing bus stops} consists of an instance $\ins = \langle N,V,b, (\theta_i)_{i\in N}\rangle$ of the base model together with a set $\exstops\subseteq \mathbb Q$ of existing bus stops.
The cost of a solution $S\subseteq V$ for agent~$i$ is then computed as
\begin{equation*}
c_i^{\ins}(S):= \min 
\begin{cases}
d(\ell_i,r_i) \\ 
\min_{x,y\in S\cup \exstops} \left[d(\ell_i,x) + \alpha\cdot d(x,y) + d(y, r_i)\right] \text. 
\end{cases}
\end{equation*}

This case can be solved by a simple modification of the dynamic program in \Cref{thm:optimal}.
Since  we can still update a cell in time $O(m)$, the running time is the same as in \Cref{thm:optimal}.

\begin{restatable}{theorem}{ThmExistingStops}
\label{thm:existingstops}
    For {\bst} with existing bus stops, we can compute the minimum cost in time $O(nm^2)$.
\end{restatable}

As a third extension, we consider the case where bus stops do not have identical costs: indeed, construction costs may vary depending on, e.g., ease of access.
An instance of {\bst} \emph{with bus stop costs} consists of an instance $\ins = \langle N,V,b, (\theta_i)_{i\in N}\rangle$ of the base model together with a \emph{budget function} 
$\gamma\colon V\to \mathbb N$.
A solution $S\subseteq V$ is \emph{feasible} if $\sum_{i\in S}\gamma(i) \le b$.
Clearly, we can still apply the dynamic program in \Cref{thm:optimal}.
However, we can no longer assume that $b\le m$.
Hence, we obtain a running time of $O(nm^2 + m^3b)$, which is only pseudo-polynomial due to the dependency on $b$.
If the bus stop costs are represented in unary, the running time remains polynomial.
However, for bus stop costs represented in binary, we 
obtain a computational hardness result, via
a reduction from {\knap}.

\begin{restatable}{proposition}{BinaryHard}\label{prop:BinaryHard}
    Let $\alpha \in [0,1)$.
    Then, the following decision problem is \NP-complete: given an instance $\ins$ of {\bst} with bus stop costs represented in binary and a rational number $q\in \mathbb Q$, decide if there exists a feasible solution $S$ with $c^{\ins}(S)\le q$. 
\end{restatable}

We conclude this section with a structural result regarding minimum-cost solutions.
Interestingly, as long as both terminals of all agents belong
to the set of potential bus stops, there is a minimum-cost solution
that places all stops at the agents' terminals.
We remark that this is the only place in this section where the assumption $\agstops\subseteq V$ plays a crucial role.

\begin{restatable}{proposition}{PropStructure}\label{prop:structure}
    For every $\alpha \in [0,1]$ and every instance $\ins$ of {\bst}
    there exists a minimum-cost feasible solution $S^*$ with $S^* \subseteq \agstops$.
\end{restatable}

    An interesting consequence of \Cref{prop:structure} is that it enables us to deal with yet another variant of the base model, where we allow infinitely large sets of potential bus stops (e.g., intervals of $\mathbb Q$).
    Indeed, we can then transform an instance $\ins$ by setting $V:=\agstops$ and apply \Cref{thm:optimal}.
In particular, this covers the case where $V = \mathbb Q$,
which can be viewed as a continuous version of our model.


\section{Fairness}\label{sec:fair}
We now turn to the consideration of fairness.
Our main contribution is an algorithm that efficiently computes outcomes that provide JR if $\alpha=0$, i.e., if taking the bus has zero cost. Moreover, the solutions computed by this algorithm lie in the $2$-approximate core (and the bound of $2$ is tight).
Besides these theoretical guarantees, we establish that our algorithm has good empirical performance: in more than 99.9\% of our (synthetically generated) instances, the algorithm finds an outcome in the core, even for $\alpha>0$.

\subsection{Theoretical Possibilities and Limitations}

We first show that JR solutions exist if taking the bus has zero cost.
For this, we consider \Cref{alg:PF0bst}.
Its key idea is to order all terminal points and select them iteratively from left to right whenever we have passed sufficiently many agent terminals (counted with multiplicities).
This approach is similar to the \textsc{CommitteeCore} algorithm by \citet{PiSk22a},
which is used to find outcomes in the core of $1$-dimensional multi-winner elections.

Given an instance $\ins = \langle N,V,b, (\theta_i)_{i\in N}\rangle$ of \bst[0],
\Cref{alg:PF0bst} 
    starts by sorting $V = \{v_1,\dots, v_m\}$ so that $v_1\le v_2\le \dots \le v_m$. 
    Then, for each $j\in [m]$, the algorithm computes $x_j$ as the number of agent terminals at or to the left of $v_j$. 
Next, for $k \in [b]$, it computes $s_k$ as the leftmost element  $v_j\in V$ with $x_j \geq k \left\lfloor\frac{2n}{b}\right\rfloor$; it then 
returns $S = \{s_k: k\in [b]\}$. 
Since $|S| \le b$, $S$ is a feasible solution.

The proof of correctness of \Cref{alg:PF0bst} is based on the following technical lemma.

\begin{restatable}{lemma}{lemDelta}
\label{lem:delta}
    Let $\ins = \langle N,V,b, (\theta_i)_{i\in N}\rangle$ be an instance of \bst[0], and let $S$ be a
    solution for $\ins$.
    Consider a coalition $M \subseteq N$ that prefers $T\subseteq V$ to $S$, and an agent $i\in M$.  
    Suppose that, when $T$ is built, $i$ walks from their left terminal $\ell_i$ to $\ell' \in T$, then takes the bus from $\ell'$ to $r' \in T$, and then walks from $r'$ to their right terminal $r_i$.
    If there exists an $\ell \in S$ with $\ell_i\le \ell\le \ell'$ or $\ell'\le \ell\le \ell_i$, then for every stop $r \in S$ it holds that $ d(\ell, \ell') < d(r, r') .$
    Similarly, if there exists a stop $r \in S$ with $r_i\le r\le r'$ or $r'\le r\le r_i$, then for every stop $\ell \in S$ it holds that $ d(\ell, \ell') > d(r, r') .$
\end{restatable}
\begin{proof}
Assume that $M$, $S$, $T$, $i$, $\ell'$ and $r'$ are as in the statement of the lemma and that there exists an $\ell\in S$ between $\ell_i$ and $\ell'$.
Fix an $r^*$ in $\argmin_{x \in S} d(r_i, x)$.
By assumption, we have $c_i(T) < c_i(S)$.
Hence, for every $r\in S$, we conclude that 
\begin{align}
    d(\ell_i, \ell') + d(r_i, r') &< d(\ell_i, \ell) + d(r_i, r^*) \qquad\text{ and hence }\nonumber\\ 
    d(\ell_i, \ell') - d(\ell_i, \ell) &< d(r_i, r^*) - d(r_i, r').\label{eq:d} 
\end{align}
Since $\ell$ lies between $\ell_i$ and $\ell'$, we have $d(\ell, \ell')=d(\ell_i, \ell')-d(\ell_i, \ell)$; substituting this into \eqref{eq:d}, we obtain
\begin{align*}
    d(\ell, \ell') <  d(r_i, r^*) - d(r_i, r') 
    \le  d(r_i, r) - d(r_i, r')
      \le d(r, r'), 
\end{align*}
where the first transition is by the choice of~$r^*$ and the second transition uses the triangle inequality.
The proof for the second statement of the lemma is analogous.
\end{proof}

\begin{algorithm}[tb]
  \caption{Fair solutions for \bst.}
  \label{alg:PF0bst}
  \begin{flushleft}
    \textbf{Input:} Instance $\ins = \langle N,V,b, (\theta_i)_{i\in N}\rangle$ of \bst\\
    \textbf{Output:} A solution $S$ 
  \end{flushleft}

  \begin{algorithmic}[]
\STATE Sort $V = \{v_1,\dots, v_m\}$ so that $v_1\le v_2\le \dots \le v_m$
  \FOR{$j = 1,\dots, m$}
  \STATE $x_j \leftarrow |\{i \in N \colon \ell_i \le v_j\}| +  |\{i \in N \colon r_i \le v_j\}|$
  \ENDFOR
  \FOR{$k = 1,\dots, b$}
  \STATE $s_k \leftarrow \min\{v_j \colon x_j \geq k \left\lfloor\frac{2n}{b} \right\rfloor\}$
  \ENDFOR
    \STATE $S \leftarrow \{s_k: k\in [b]\}$
  \RETURN $S$
 \end{algorithmic}
\end{algorithm}

\begin{restatable}{theorem}{propfair}\label{thm:propfair}
    \Cref{alg:PF0bst} runs in polynomial time, and for \bst[0] it computes feasible solutions that provide JR.
\end{restatable}

\begin{proof}
Clearly, \Cref{alg:PF0bst} runs in polynomial time.
We claim that the solution $S$ computed by \Cref{alg:PF0bst} provides JR. 

Indeed, assume for contradiction
that there exists a coalition $M \subseteq N$ of size $|M| \geq 2\cdot\frac nb$ and a pair of stops $T =\{\ell', r'\} \subseteq V$ such that $c_i(T) < c_i(S)$ for all $i\in M$. 

Let $L=\{\ell_i:i\in M\}$. Since there are at most $\left\lfloor \frac{2n}{b}\right\rfloor$ terminals between every two consecutive stops in $S$, there exists an $\ell^* \in S$ such that at least one of the terminals in $L$ is before $\ell^*$ or exactly at $\ell^*$, and at least one of the terminals in $L$ is after $\ell^*$ or exactly at $\ell^*$. 
Hence, if $\ell^*\le \ell'$, then there exists an agent $i \in M$ with $\ell_i\le \ell^*\le \ell'$ and if $\ell'\le \ell^*$, then there exists an agent $i \in M$ with $\ell'\le \ell^*\le \ell_i$. 
By \Cref{lem:delta} for every $r \in S$ we have $d(\ell^*, \ell') < d(r, r')$.

By a similar argument, there exists an agent $j \in M$ and a stop $r^* \in S$ where $r_j\le r^*\le r'$ or $r'\le r^*\le r_j$. 
By \Cref{lem:delta} for every $\ell \in S$ we have $d(r^*, r') < d(\ell, \ell')$.
Setting $r = r^*$ and $\ell = \ell^*$, we obtain a contradiction.
\end{proof}

\begin{example}
    Consider the execution of \Cref{alg:PF0bst} on the instance from \Cref{ex:incompatible}.
    There, we had $V = \{0,1,4,7,10,13,15\}$, leading to the values $(x_j)_{j=1}^7 = (2,6,7,8,9,10,12)$.
    This leads to $(s_k)_{k=1}^6 = (0,1,1,7,13,15)$.
    Hence, \Cref{alg:PF0bst} returns $S = \{0,1,7,13,15\}$ for this instance.
    \hfill$\lhd$
\end{example}

The example shows that, while \Cref{alg:PF0bst} always returns a feasible solution that provides JR, the obtained solution may fail to exhaust the budget.
In such cases, having achieved JR via \Cref{alg:PF0bst}, we can distribute the remaining budget to accomplish other goals.
For instance, as per \Cref{thm:existingstops}, we can extend the solution to lower the total cost as much as possible.
Note that, after we add stops, the solution continues to provide JR.

A natural follow-up question is whether \Cref{thm:propfair} can be extended to arbitrary $\alpha \in (0,1)$. Unfortunately, our next result shows that this is not the case.

\begin{restatable}{proposition}{PFalphaNotZero}\label{prop:PFalphaNotZero}
    Let $\alpha \in (0,1)$. 
    Then, for \bst, \Cref{alg:PF0bst} may return a solution that does not provide JR.
\end{restatable}

\begin{proof}
    Let $\alpha\in (0,1)$.
    We set $\lambda = 2\frac{1+\alpha}{1-\alpha} + 1$, and
    define an instance 
    $\ins = \langle N,V,b, (\theta_i)_{i\in N}\rangle$  
     with $N = \{1, 2, 3, 4\}$, $b=4$,  
     $V=\{0, \lambda+4\}\cup\{\lambda+j, -\lambda-j: j=0, 1, 2, 3\}$,
    and agent types given by
    $$
    \theta_1 = (-3-\lambda,\lambda), 
    \theta_2 = (-\lambda,\lambda + 3), 
    \theta_3 = (-2-\lambda,\lambda + 2),
    \theta_4 = (0,\lambda+4).
    $$
    This instance is illustrated in \Cref{fig:PFalphaNot0}.
    Note that agents~$1$ and~$2$ have symmetric roles with respect to the bus stop at $0$.
    Moreover, the budget is equal to the number of agents.
    
\begin{figure*}
    \centering
    \resizebox{1\textwidth}{!}{
    \begin{tikzpicture}
        \draw[->] (-.7,0) -- (14.4,0);
        \foreach[count = \k] \i in {0,1.3,2.6,3.9,6.2,8.5,9.8,11.1,12.4,13.7}
        {
        \draw (\i,-.1) -- (\i,.1);
        \node (p\k) at (\i,-.4) {};
        }
        \foreach[count = \k] \i in {-3-\lambda,-2-\lambda,-1-\lambda,-\lambda,0,\lambda,\lambda+1,\lambda+2,\lambda+3,\lambda+4}
        {
        \node at (p\k) {$\i$};
        }
        \node[align = center] at (-1.5,-.4) {terminals};
        \node[align = center] at (-1.5,.9) {desired\\routes};
        \draw ($(p1) + (0,.7)$) edge[->] node[pos =.5, fill = white] {$a_1$} ($(p6) + (0,.7)$);
        \draw ($(p4) + (0,1)$) edge[->] node[pos =.5, fill = white] {$a_2$} ($(p9) + (0,1)$);
        \draw ($(p2) + (0,1.3)$) edge[->] node[pos =.5, fill = white] {$a_3$} ($(p8) + (0,1.3)$);
        \draw ($(p5) + (0,1.6)$) edge[->] node[pos =.5, fill = white] {$a_4$} ($(p10) + (0,1.6)$);
    \end{tikzpicture}
    }\vspace{-2mm}
    \caption{\Cref{alg:PF0bst} fails to output JR solutions for $\alpha \in (0,1)$.}
    \label{fig:PFalphaNot0}
\end{figure*}
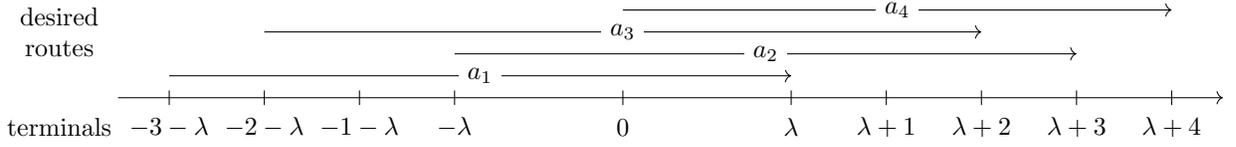

On $\ins$, \Cref{alg:PF0bst} outputs $S = \{-2-\lambda,0,\lambda+2,\lambda+4\}$.
We claim that the set of agents $M = \{a_1,a_2\}$ and the set of stops $T = \{-1-\lambda,\lambda+1\}$.
form a certificate that $S$ does not provide JR.

If stops in $T$ are built, agents $1$ and $2$ can use the bus between these two stops, so we have $c_1(T), c_2(T) \le 3 + \alpha(2\lambda+2)$. 

Now, consider their costs under $S$.
Agent $1$ has to start by walking for one step.
Then, between $-2-\lambda$ and $0$ it is fastest to take the bus.
Finally, there are two options for traveling from $0$ to $\lambda$, the destination of agent~$1$.
Taking the bus and walking back has a cost of $\alpha(\lambda + 2)+2 = 2(1+ \alpha) + \alpha\lambda$, whereas walking has a cost of $\lambda$.
Note that $\lambda > 2(1+ \alpha) + \alpha\lambda$ if and only if $\lambda > 2\frac{1+\alpha}{1-\alpha}$, which is exactly how we chose $\lambda$.
Hence, taking the bus and then walking back is faster.
Thus, agent~$1$ will walk from $-3-\lambda$ to $-2-\lambda$, take the bus to $\lambda+2$ and walk back to $\lambda$, which has a cost of $c_1(S) = 3 + \alpha(2\lambda + 4) > c_1(T)$.
A similar computation shows that $c_2(S) = 3 + \alpha(2\lambda + 4) > c_2(T)$.
Together, this shows that $S$ fails to provide JR.
\end{proof}

We can show further
fairness guarantees for \Cref{alg:PF0bst}, namely
that it computes solutions in the $2$-core. 
However, the approximation guarantee
of~$2$ is tight.

\begin{restatable}{theorem}{approxBound}
\label{prop:approxbound}
    Let $\alpha \in [0,1)$. Then, for \bst, \Cref{alg:PF0bst} computes solutions in the $2$-core.
    However, for each $0 < \epsilon \le 1$, it may output solutions that are not in the $(2-\epsilon)$-core.
\end{restatable}

\begin{proof}
    Let $\alpha \in [0,1)$.
    We start by proving that the output of \Cref{alg:PF0bst} is in the $2$-core.
Consider an instance $\ins = \langle N,V,b, (\theta_i)_{i\in N}\rangle$ and a solution $S$ computed by \Cref{alg:PF0bst}
on $\ins$. 
Assume for contradiction that there is a set of agents $M\subseteq N$ and a set of stops $T\subseteq V$ with $|T|\le \frac {b}{2n}\cdot|M|$ such that $c_i(T) < c_i(S)$ for all $i\in M$.

For each $t\in T$, define $\ell^t := \max(\{s\in S\colon s\le t\}\cup\{-\infty\})$ and $r^t := \min(\{s\in S\colon s\ge t\}\cup\{\infty\})$.
These are the closest bus stops in $S$ to the left and to the right of $t$; if there is no such stop, we set this variable to $-\infty$ or $\infty$, respectively.
Let $C_t := \{i\in N\colon \ell^t < \ell_i < r^t \text{ or } \ell^t < r_i < r^t\}$, i.e., $C_t$ is the set of agents that have at least one terminal strictly between $\ell^t$ and $r^t$.
By design of \Cref{alg:PF0bst}, we have
\begin{equation}\label{eq:Cbound}
    |C_t|\le \frac {2n}b - 1\text.
\end{equation}

We claim that for every $i\in M$ there exists $t\in T$ with $i\in C_t$.
Indeed, since $c_i(T) < c_i(S)$, we know that the cost of $i$ with respect to $T$ comes from a route in which they take the bus, i.e., there exist stops $t_1, t_2\in T$ such that $i$ walks from $\ell_i$ to $t_1$, takes the bus to $t_2$, and then walks to $r_i$.
If $i\notin C_{t_1}\cup C_{t_2}$, then there exist $s_1,s_2\in S$ such that $s_1$ is between $\ell_i$ and $t_1$, and $s_2$ is between $t_2$ and $r_i$.
It is easy to see that $c_i(T) = c_i(\{t_1,t_2\})\ge c_i(\{s_1,s_2\})\ge c_i(S)$, contradicting the cost improvement of $i$ according to $T$.
Hence, it follows that $i\in C_{t_1}\cup C_{t_2}$.

We conclude that
$$
|M|\le \left|\bigcup_{t\in T}C_t\right|\le \sum_{t\in T}|C_t| \overset{\text{Eq.~(\ref{eq:Cbound})}}{\le} |T|\left(\frac {2n}b - 1\right) < |M|\text.
$$

This is a contradiction, and hence, $S$ is in the $2$-core.

    To prove that the bound is tight, let $0 < \epsilon \le 1$.
    We define an instance that is parameterized by an integer $k$.
    First, observe that $\frac {k+1}{2k}(2-\epsilon)$ converges to $1-\frac{\epsilon}2$ as $k$ tends to infinity.
    Hence, we can choose $k$ large enough so that $\frac {k+1}{2k}(2-\epsilon) \le 1 - \frac{\epsilon}4$.
    Moreover, fix an integer 
    $x > \frac 4{\epsilon}$.
    We define an instance with $b = 2k$ and $k$ groups of agents $N^1, \dots, N^k$ of size~$x$ each, so $n=kx$.
    We want the agents' left and right terminals to be separated by a sufficiently long part of the path with no agent terminals. 
    To this end, for each $j\in [k]$, we let $\ell^j=j$, $r^j=2k+(2j-1)k$, and set
    \begin{itemize}
        \item $\ell_i:=\ell^j$ for each $i\in N^j$;
        \item $r_i:=r^j$ for $x-1$ agents in $N^j$ and $r_i=r^j+k$ for the remaining agent in $N^j$.
    \end{itemize}
    An illustration is given in \Cref{fig:approxbound}.
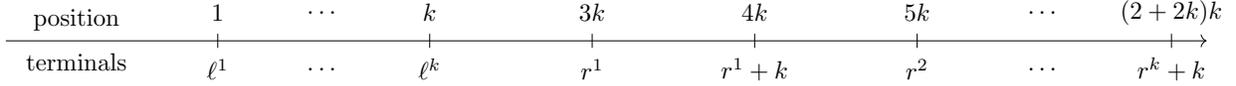
\begin{figure*}[tb] 
        \centering
            \resizebox{1\textwidth}{!}{
            \begin{tikzpicture}
                \draw[->] (-2,0) -- (15,0);
                \node at (-1,.3) {position};
                \node at (-1,-.3) {terminals};

                \foreach \i in {1,4,6.3,8.6,10.9,14.5}
                {\draw (\i,.1) -- (\i,-.1);}
                
                \node (1) at (1,.4) {$1$};
                \node (k) at (4,.4) {$k$};
                \node at ($(1)!.5!(k)$) {\dots};
                \node (2k) at (6.3,.4) {$3k$};
                \node (3k) at (8.6,.4) {$4k$};
                \node (4k) at (10.9,.4) {$5k$};
                \node (kk) at 
                (14.5,.4) {$(2+2k)k$};
                \node at ($(4k)!.5!(kk)$) {\dots};
                
                \node (1) at (1,-.4) {$\ell^1$};
                \node[align = center] (k) at (4,-.4) {$\ell^k$};
                \node at ($(1)!.5!(k)$) {\dots};
                \node (2k) at (6.3,-.4) {$r^1$};
                \node (3k) at (8.6,-.4) {$r^1+k$};
                \node[align = center] (4k) at (10.9,-.4) {$r^2$};
                \node (kk) at (14.5,-.4) {$r^k+k$};
                \node at ($(4k)!.5!(kk)$) {\dots};
            \end{tikzpicture}
            } 
        \caption{
        \Cref{alg:PF0bst} fails to output solutions in the $(2-\epsilon)$-core for $\epsilon \in (0,1]$.}
        \label{fig:approxbound}
    \end{figure*}
    For $V= \cup_{i \in N} \{\ell_i, r_i\}$, the algorithm outputs the $2k$ bus stops in $S = [k]\cup \{2k+ 2jk\colon j\in [k]\}$.

    Now, consider the coalition $M$ consisting of all agents except those whose right terminals are at $r^j+k$, $j\in [k]$. We have 
    \begin{align*}
        |M| &= n-k = k(x-1) > k\left(x - \frac{\epsilon}4x\right) = \left(1 - \frac{\epsilon}4\right)kx \\
               &\ge \frac{k+1}{2k}\left(2-\epsilon\right)kx = \left(k+1\right)\left(2-\epsilon\right)\frac nb\text.
    \end{align*}
    
    Here, the first inequality holds by our choice of $x$, and the second inequality holds by our choice of $k$.
    Consider the set of bus stops 
    $$
    T = \{k\}\cup\{2k + (2j-1)k\colon j\in [k]\},
    $$
    with $|T| = k+1$.
    To show that $M$ is a blocking coalition for the $(2-\epsilon)$-core, it remains to argue that for every agent $i\in M$ it holds that $c_i(T) < c_i(S)$.
    To see this, fix a $j\in [k]$ and $i\in M\cap N^j$.
    Then, given solution $T$, $i$ can walk to $k$ and take the bus all the way to $r_i$.
    However, in $S$, $i$ can board the bus at $\ell_i$, but then she needs to walk at least $k$ to her destination. Since the overall distance she travels in the latter case is at least as large as in the former case, and involves strictly more walking ($k$ instead of $k-j$), agent $i$ prefers $T$ to $S$.
\end{proof}

We note that \bst{} can be shown to be a special case of committee selection with monotonic preferences and uniform costs. 
Hence, the results of \citet{JiangMW20} (Theorem~1) imply that the $16$-core of \bst{} is always non-empty. 
However, Theorem~\ref{prop:approxbound} offers a much better approximation guarantee.

It remains an open problem how to construct solutions in the core;
in fact, we do not even know if the core is always non-empty.
This seems to be a challenging question. 
For instance, by \Cref{thm:effVSfair}, the cost-minimal solution does not even provide JR.
Similarly, the solution that selects the most popular terminal 
points may not provide JR; an example is given in \Cref{app:further}.

Finally, we observe that solutions that provide strong JR (and therefore lie in the strong core) need not exist.

\begin{restatable}{proposition}{strongJR}\label{prop:strongJR}
    For every $\alpha \in [0, 1)$ there exists an instance of {\bst} such that no feasible solution provides strong JR.
\end{restatable}
\begin{proof}
Consider an instance $\ins = \langle N,V,b, (\theta_i)_{i\in N}\rangle$ with agent set $N = [8]$ and budget $b = 4$.
    We set $V = [16]$, and for each $i\in N$ we set $\ell_i = 2i-1$ and $r_i = 2i$.
    Hence, we have an instance with~$8$ agents, where all agents have different terminals.

    Then, any feasible solution $S$ has an empty intersection with the set of terminals of at least $4$ agents. 
    Hence, there exists a set $M\subseteq N$ with $|M|\ge 4$ such that $c_i(S) = 1$ for all $i\in M$, i.e., no agent in $M$ can do any better than to walk. 
    Note that $|M| \ge \frac{2n}b$.
    Hence, the agents in $M$ are entitled to two bus stops.

    Now, let $i\in M$ and consider $T = \{\ell_i, r_i\}$.
    Then $c_i(T) = \alpha < 1 = c_i(S)$, and for $j \in M \setminus \{i\}$ it holds that $c_j(T) = 1 = c_j(S)$. 
    Hence, $T$ is strictly better for agent~$i$ and at least as good for all other members of $M$. 
    Therefore, $S$ does not provide strong JR.
\end{proof}
\subsection{Computation of Outcomes in the Core on Synthetic Data}
\label{sec:experiments}

While it remains open whether the core of \bst{} is always non-empty, 
our simulations indicate that \Cref{alg:PF0bst} can frequently find outcomes in the core. 
This is despite the fact that, according to \Cref{prop:PFalphaNotZero} and
\Cref{prop:approxbound}, in the worst case this algorithm may fail to find JR outcomes for $\alpha \in (0,1)$ and outcomes in the $\appr$-core for $\appr<2$. 
Computations were performed using an Apple M2 CPU with 24 GB of RAM.
All source code is publicly available at \url{https://github.com/latifian/Bus_Stop_Model}.

\subsubsection{Experimental Setup}

For our experiments, we 
consider the following parameter ranges:
\begin{itemize}
    \item Number of agents $n\in \{5, \ldots, 25\}$.
    \item Number of bus stops $m\in \{5, \ldots, 15\}$.
    \item Budget $b\in \{3, \ldots, m-1\}$, i.e., we consider all budgets that facilitate taking the bus at all and do not enable building all possible bus stops.
    \item Cost parameter $\alpha \in \{0, 0.1, \ldots, 0.9\}$.
\end{itemize}
For each combination, we generated 400 random instances.

\paragraph{Generating Random Instances}
Our set of bus stops is a subset of the integers between $1$ and $100$.
From this range, we choose $m$ potential locations for bus stops uniformly at random.
Moreover, we determine the types of agents by independently selecting two terminals from the set of possible stops uniformly at random.

\begin{figure*}[h!]
 \centering
\begin{minipage}[t]{0.49\textwidth}
    \centering
    \includegraphics[height=6.8cm]{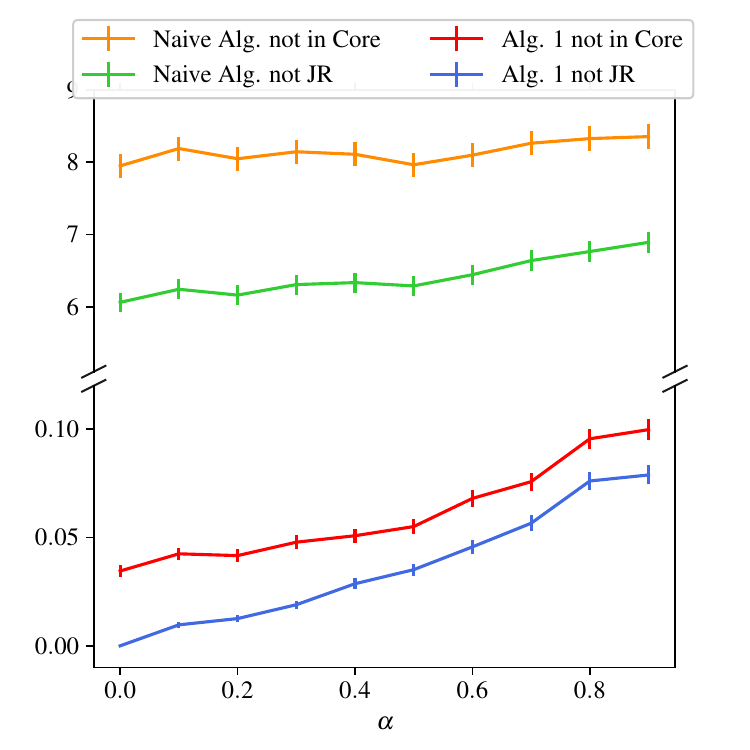}
    \caption{Aggregated frequency of fairness violations of the solutions computed by \Cref{alg:PF0bst} and the naive algorithm along with the standard error. The $x$-axis shows our range for the cost parameter $\alpha$ and the $y$-axis shows the percentage of the instances in which the desired property is not satisfied.}
    \label{fig:frequency}
\end{minipage}
\hfill
\begin{minipage}[t]{0.49\textwidth}
    \centering
    \includegraphics[height=6.8cm]{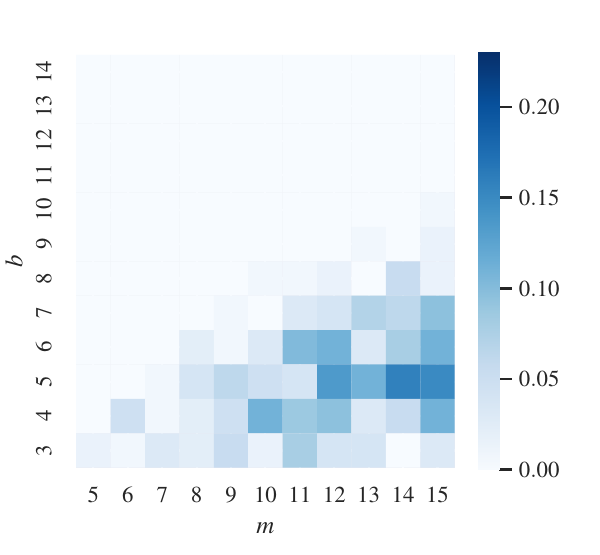}
    \caption{A heat map showing the frequency of core violations (as a percentage, i.e., 0.05 means 0.05\%) of the solutions computed by \Cref{alg:PF0bst} for pairs of $m$ and $b$ in instances of \bst[0]. 
    Each cell is averaged over all values of $n$.}
    \label{fig:avg_heatmap}
\end{minipage}
   
\end{figure*}
\paragraph{Algorithmic Benchmark}
We measure the frequency of instances on which \Cref{alg:PF0bst} computes solutions that provide JR/are in the core.
For comparison, we benchmark our algorithm against the naive algorithm, which selects the stops as if the agents were distributed uniformly on the line.
More precisely, if $V$ is the set of possible stops, the naive algorithm ignores the actual agents' types and assumes instead 
that for each point $p \in \{1, \ldots, 100\}$, $n/50$ agents have one of their endpoints at $p$ and report a point in $\arg\min d(p, V)$ as the respective terminal;
it then runs \Cref{alg:PF0bst} under this assumption on the agents' types. 
In other words, \Cref{alg:PF0bst} can be seen as a weighted version of the naive algorithm, which takes actual user demands into account.
In \Cref{app:benchmark2}, we consider a second benchmark algorithm, which is based on the idea of selecting the most demanded bus stops.

\paragraph{Verification of Fair Solutions}
An integral part of performing simulations is an efficient algorithm for verifying whether a given solution provides JR or is in the core.
While JR can be checked by a simple polynomial-time algorithm (by checking if enough agents benefit from using each pair of stops), we set up an integer program to verify whether a solution is in the core.
The theory for this part of the simulations is described in \Cref{app:verification}.

\subsubsection{Experimental Results}

The primary goal of our experiments is to measure the performance of \Cref{alg:PF0bst} in terms of achieving outcomes in the core or providing JR beyond the guarantee of \Cref{thm:propfair}.
\Cref{fig:frequency} shows an aggregated view of our results.
We see that, for a fixed $\alpha$, \Cref{alg:PF0bst} computes a solution in the core in more than $99.9\%$ of the cases.
Even for the parameter combinations with the highest frequency of fairness violations, the failure rate of \Cref{alg:PF0bst} with respect to computing solutions in the core does not exceed 3\%.
\Cref{fig:avg_heatmap} gives a glimpse at a more nuanced distribution of the failures of computing solutions in the core.
We defer a more detailed analysis to \Cref{app:DetailedSimu}.
We remark that the upper left triangle of the figure cannot contain any core violations because then the budget exceeds the number of stops. 
One trend that can be observed from \Cref{fig:avg_heatmap} and that is confirmed in our analysis for any fixed number of agents is that the frequency of core violations is the highest for a comparatively high number of potential bus stops and a smaller budget.

Another interesting observation is that, for a significant fraction of the instances, if the solution computed by \Cref{alg:PF0bst} is not in the core, it fails JR as well;
indeed, this fraction tends to increase as the cost parameter $\alpha$ increases.
Moreover, \Cref{alg:PF0bst} performs much better than the naive algorithm.
On average it performs up to $230.2$ better for smaller values of $\alpha$ and still $83.7$ times better 
for $\alpha = 0.9$.
A detailed comparison of the performance of the two algorithms is provided in \Cref{tab:comparison:algos} in \Cref{app:DetailedSimu}.
Our interpretation of this finding is that, while placing bus stops uniformly is a simple and appealing approach, taking into account the actual user demands results in much fairer solutions.

\section{Conclusion}

We proposed a stylized model for planning a bus route.
Our model can capture efficiency in terms of travel costs
as well as fairness in terms of proportional representation of the agents.
We have developed a dynamic program that minimizes the total travel cost for the agents.
This approach turned out to be extremely versatile, 
in that it also applies to many variants of the base model.
Concerning fairness, our main contribution is an algorithm that constructs JR solutions under the assumption that taking the bus has no cost.
This algorithm is also a $2$-approximation for the core.

Our work suggests two natural avenues for further research.
First, it remains open how to compute JR solutions
for instances of {\bst} where $\alpha \neq 0$; indeed, we do not even know if such solutions always exist.
An even harder question is whether the core is always non-empty.
This resembles the situation in approval-based committee voting, where the same question is famously open \citep[see, e.g.,][]{LaSk22b}.
Another promising direction is to extend our model to more complex topologies.
For instance, one can consider the setting where the set of potential stops 
(and terminals) is a subset of $\mathbb Q^2$ or the vertex set of a (planar) graph.
Further ahead, an important research challenge is to develop a richer framework to reason about fairness in more realistic models of public transport.


\section*{Acknowledgements}

    Most of this work was done while Edith Elkind was at University of Oxford. 
    Martin Bullinger was supported by the AI Programme of The Alan Turing Institute. 
    Edith Elkind and Mohamad Latifian were supported by the UK Engineering and Physical Sciences Research Council (EPSRC) grants EP/X038548/1 and EP/Y003624/1.
    We would like to thank the anonymous reviewers from AAMAS.
    
\bibliography{busstop}

\begin{thebibliography}{42}
\providecommand{\natexlab}[1]{#1}
\providecommand{\url}[1]{\texttt{#1}}
\expandafter\ifx\csname urlstyle\endcsname\relax
  \providecommand{\doi}[1]{doi: #1}\else
  \providecommand{\doi}{doi: \begingroup \urlstyle{rm}\Url}\fi

\bibitem[Agarwal et~al.(2022)Agarwal, Ko, Munagala, and Taylor]{AKMT22a}
Pankaj~K. Agarwal, Shao-Heng Ko, Kamesh Munagala, and Erin Taylor.
\newblock Locally fair partitioning.
\newblock In \emph{Proceedings of the AAAI Conference on Artificial
  Intelligence (AAAI)}, pages 4752--4759, 2022.

\bibitem[Anastasiadis and Deligkas(2018)]{AnastasiadisD18}
Eleftherios Anastasiadis and Argyrios Deligkas.
\newblock Heterogeneous facility location games.
\newblock In \emph{Proceedings of the 17th International Conference on
  Autonomous Agents and Multiagent Systems (AAMAS)}, pages 623--631, 2018.

\bibitem[Aziz et~al.(2017)Aziz, Brill, Conitzer, Elkind, Freeman, and
  Walsh]{aziz2017justified}
Haris Aziz, Markus Brill, Vincent Conitzer, Edith Elkind, Rupert Freeman, and
  Toby Walsh.
\newblock Justified representation in approval-based committee voting.
\newblock \emph{Social Choice and Welfare}, 48\penalty0 (2):\penalty0 461--485,
  2017.

\bibitem[Aziz et~al.(2023)Aziz, Micha, and Shah]{AMS23a}
Haris Aziz, Evi Micha, and Nisarg Shah.
\newblock Group fairness in peer review.
\newblock In \emph{Proceedings of the 36th Annual Conference on Neural
  Information Processing (NeurIPS)}, 2023.

\bibitem[Aziz et~al.(2024)Aziz, Lee, Chu, and Vollen]{aziz24}
Haris Aziz, Barton~E. Lee, Sean~Morota Chu, and Jeremy Vollen.
\newblock Proportionally representative clustering.
\newblock In \emph{Proceedings of the 20th International Conference on Web and
  Internet Economics (WINE)}, page to appear, 2024.

\bibitem[Chan and Wang(2023)]{ChWa23a}
Hau Chan and Chenhao Wang.
\newblock Mechanism design for improving accessibility to public facilities.
\newblock In \emph{Proceedings of the 22nd International Conference on
  Autonomous Agents and Multiagent Systems (AAMAS)}, pages 2116--2124, 2023.

\bibitem[Chan et~al.(2021)Chan, Filos{-}Ratsikas, Li, Li, and Wang]{ChanFLLW21}
Hau Chan, Aris Filos{-}Ratsikas, Bo~Li, Minming Li, and Chenhao Wang.
\newblock Mechanism design for facility location problems: {A} survey.
\newblock In \emph{Proceedings of the 30th International Joint Conference on
  Artificial Intelligence (IJCAI)}, pages 4356--4365, 2021.

\bibitem[Chapman(2007)]{Chap07a}
Lee Chapman.
\newblock Transport and climate change: a review.
\newblock \emph{Journal of transport geography}, 15\penalty0 (5):\penalty0
  354--367, 2007.

\bibitem[Chaudhury et~al.(2022)Chaudhury, Li, Kang, Li, and
  Mehta]{ChaudhuryLK0M22}
Bhaskar~Ray Chaudhury, Linyi Li, Mintong Kang, Bo~Li, and Ruta Mehta.
\newblock Fairness in federated learning via core-stability.
\newblock In \emph{Proceedings of the 35th Annual Conference on Neural
  Information Processing (NeurIPS)}, 2022.

\bibitem[Chaudhury et~al.(2024)Chaudhury, Murhekar, Yuan, Li, Mehta, and
  Procaccia]{ChaudhuryMY0MP24}
Bhaskar~Ray Chaudhury, Aniket Murhekar, Zhuowen Yuan, Bo~Li, Ruta Mehta, and
  Ariel~D. Procaccia.
\newblock Fair federated learning via the proportional veto core.
\newblock In \emph{Proceedings of the 41st International Conference on Machine
  Learning (ICML)}, 2024.

\bibitem[Chen et~al.(2019)Chen, Fain, Lyu, and Munagala]{CFLM19a}
Xingyu Chen, Brandon Fain, Liang Lyu, and Kamesh Munagala.
\newblock Proportionally fair clustering.
\newblock In \emph{Proceedings of the 36th International Conference on Machine
  Learning (ICML)}, pages 1032--1041, 2019.

\bibitem[Conitzer et~al.(2017)Conitzer, Freeman, and Shah]{ConitzerF017}
Vincent Conitzer, Rupert Freeman, and Nisarg Shah.
\newblock Fair public decision making.
\newblock In \emph{Proceedings of the 18th {ACM} Conference on Economics and
  Computation (ACM EC)}, pages 629--646, 2017.

\bibitem[Daduna and Pinto~Paix{\~a}o(1995)]{DaPa95a}
Joachim~R. Daduna and Jos{\'e}~M. Pinto~Paix{\~a}o.
\newblock Vehicle scheduling for public mass transit--an overview.
\newblock In \emph{Computer-Aided Transit Scheduling: Proceedings of the 6th
  International Workshop on Computer-Aided Scheduling of Public Transport},
  pages 76--90, 1995.

\bibitem[Desaulniers and Hickman(2007)]{DeHi07a}
Guy Desaulniers and Mark~D. Hickman.
\newblock Public transit.
\newblock \emph{Handbooks in operations research and management science},
  14:\penalty0 69--127, 2007.

\bibitem[Elkind et~al.(2022)Elkind, Li, and Zhou]{ElkindLZ22}
Edith Elkind, Minming Li, and Houyu Zhou.
\newblock Facility location with approval preferences: Strategyproofness and
  fairness.
\newblock In \emph{Proceedings of the 21st International Conference on
  Autonomous Agents and Multiagent Systems (AAMAS)}, pages 391--399, 2022.

\bibitem[{European Parliament}(2019)]{Emissions19a}
{European Parliament}.
\newblock Co2 emissions from cars: facts and figures.
\newblock
  \url{https://www.europarl.europa.eu/news/en/headlines/society/20190313STO31218/co2-emissions-from-cars-facts-and-figures-infographics},
  2019.
\newblock [Online; accessed 16-January-2024].

\bibitem[Fain et~al.(2016)Fain, Goel, and Munagala]{fain2016core}
Brandon Fain, Ashish Goel, and Kamesh Munagala.
\newblock The core of the participatory budgeting problem.
\newblock In \emph{Proceedings of the 12th International Conference on Web and
  Internet Economics (WINE)}, pages 384--399, 2016.

\bibitem[Forkenbrock and Schweitzer(1999)]{FoSc99a}
David~J. Forkenbrock and Lisa~A. Schweitzer.
\newblock Environmental justice in transportation planning.
\newblock \emph{Journal of the American Planning Association}, 65\penalty0
  (1):\penalty0 96--112, 1999.

\bibitem[Fukui et~al.(2020)Fukui, Shurbevski, and Nagamochi]{FSN20a}
Yuhei Fukui, Aleksandar Shurbevski, and Hiroshi Nagamochi.
\newblock Group strategy-proof mechanisms for shuttle facility games.
\newblock \emph{Journal of Information Processing}, 28:\penalty0 976--986,
  2020.

\bibitem[Gillies(1959)]{gillies1959solutions}
Donald~B Gillies.
\newblock Solutions to general non-zero-sum games.
\newblock \emph{Contributions to the Theory of Games}, 4\penalty0
  (40):\penalty0 47--85, 1959.

\bibitem[He et~al.(2024)He, Botan, Lang, Saffidine, Sikora, and
  Workman]{HBL+24}
Zixu He, Sirin Botan, J{\'{e}}r{\^{o}}me Lang, Abdallah Saffidine, Florian
  Sikora, and Silas Workman.
\newblock Fair railway network design.
\newblock Technical report, arXiv abs/2409.02152, 2024.

\bibitem[Jiang et~al.(2020)Jiang, Munagala, and Wang]{JiangMW20}
Zhihao Jiang, Kamesh Munagala, and Kangning Wang.
\newblock Approximately stable committee selection.
\newblock In \emph{Proceedings of the 52nd Annual {ACM} {SIGACT} Symposium on
  Theory of Computing (STOC)}, pages 463--472, 2020.

\bibitem[Jozefowiez et~al.(2009)Jozefowiez, Semet, and Talbi]{JST09a}
Nicolas Jozefowiez, Fr{\'e}d{\'e}ric Semet, and El-Ghazali Talbi.
\newblock An evolutionary algorithm for the vehicle routing problem with route
  balancing.
\newblock \emph{European Journal of Operational Research}, 195\penalty0
  (3):\penalty0 761--769, 2009.

\bibitem[Kalayci et~al.(2024)Kalayci, Kempe, and Kher]{KKK23a}
Yusuf~Hakan Kalayci, David Kempe, and Vikram Kher.
\newblock Proportional representation in metric spaces and low-distortion
  committee selection.
\newblock In \emph{Proceedings of the 38th {AAAI} Conference on Artificial
  Intelligence (AAAI)}, pages 9815--9823, 2024.

\bibitem[Karp(1972)]{Karp72a}
Richard~M. Karp.
\newblock Reducibility among combinatorial problems.
\newblock In R.~E. Miller and J.~W. Thatcher, editors, \emph{Complexity of
  Computer Computations}, pages 85--103. Plenum Press, 1972.

\bibitem[Kwan and Hashim(2016)]{KwHa16a}
Soo~Chen Kwan and Jamal~Hisham Hashim.
\newblock A review on co-benefits of mass public transportation in climate
  change mitigation.
\newblock \emph{Sustainable Cities and Society}, 22:\penalty0 11--18, 2016.

\bibitem[Lackner and Skowron(2023)]{LaSk22b}
Martin Lackner and Piotr Skowron.
\newblock \emph{Multi-Winner Voting with Approval Preferences}.
\newblock Springer-Verlag, 2023.

\bibitem[Lampkin and Saalmans(1967)]{LaSa67a}
William Lampkin and P.~D. Saalmans.
\newblock The design of routes, service frequencies, and schedules for a
  municipal bus undertaking: A case study.
\newblock \emph{Journal of the Operational Research Society}, 18\penalty0
  (4):\penalty0 375--397, 1967.

\bibitem[Li et~al.(2021)Li, Li, Sun, Wang, and Wang]{LLS+21a}
Bo~Li, Lijun Li, Ankang Sun, Chenhao Wang, and Yingfan Wang.
\newblock Approximate group fairness for clustering.
\newblock In \emph{Proceedings of the 38th International Conference on Machine
  Learning (ICML)}, pages 6381--6391, 2021.

\bibitem[Matl et~al.(2018)Matl, Hartl, and Vidal]{MHV18a}
Piotr Matl, Richard~F. Hartl, and Thibaut Vidal.
\newblock Workload equity in vehicle routing problems: A survey and analysis.
\newblock \emph{Transportation Science}, 52\penalty0 (2):\penalty0 239--260,
  2018.

\bibitem[Micha and Shah(2020)]{MiSh20a}
Evi Micha and Nisarg Shah.
\newblock Proportionally fair clustering revisited.
\newblock In \emph{Proceedings of the 47th International Colloquium on
  Automata, Languages, and Programming (ICALP)}, 2020.

\bibitem[Pereira et~al.(2017)Pereira, Schwanen, and Banister]{PSB17a}
Rafael H.~M. Pereira, Tim Schwanen, and David Banister.
\newblock Distributive justice and equity in transportation.
\newblock \emph{Transport reviews}, 37\penalty0 (2):\penalty0 170--191, 2017.

\bibitem[Peters and Skowron(2020)]{PetersS20}
Dominik Peters and Piotr Skowron.
\newblock Proportionality and the limits of welfarism.
\newblock In \emph{Proceedings of the 21st {ACM} Conference on Economics and
  Computation (ACM EC)}, pages 793--794, 2020.

\bibitem[Pierczy{\'n}ski and Skowron(2022)]{PiSk22a}
Grzegorz Pierczy{\'n}ski and Piotr Skowron.
\newblock Core-stable committees under restricted domains.
\newblock In \emph{Proceedings of the 18th International Conference on Web and
  Internet Economics (WINE)}, pages 311--329, 2022.

\bibitem[Pucher(1982)]{Puch82a}
John Pucher.
\newblock Discrimination in mass transit.
\newblock \emph{Journal of the American Planning Association}, 48\penalty0
  (3):\penalty0 315--326, 1982.

\bibitem[Rawls(1971)]{Rawl71a}
John Rawls.
\newblock \emph{A Theory of Justice}.
\newblock Harvard University Press, 1971.

\bibitem[Sen(2009)]{Sen09a}
Amartya Sen.
\newblock \emph{The Idea of Justice}.
\newblock Belknap Press of Harvard University Press, 2009.

\bibitem[Serafino and Ventre(2014)]{SerafinoV14}
Paolo Serafino and Carmine Ventre.
\newblock Heterogeneous facility location without money on the line.
\newblock In \emph{Proceedings of the 21st European Conference on Artificial
  Intelligence (ECAI)}, pages 807--812, 2014.

\bibitem[Silman et~al.(1974)Silman, Barzily, and Passy]{SBP74a}
Lionel~Adrian Silman, Zeev Barzily, and Ury Passy.
\newblock Planning the route system for urban buses.
\newblock \emph{Computers \& operations research}, 1\penalty0 (2):\penalty0
  201--211, 1974.

\bibitem[Waterson et~al.(2003)Waterson, Rajbhandari, and Hounsell]{WRH03a}
Ben~J. Waterson, Bipul Rajbhandari, and Nick~B. Hounsell.
\newblock Simulating the impacts of strong bus priority measures.
\newblock \emph{Journal of Transportation Engineering}, 129\penalty0
  (6):\penalty0 642--647, 2003.

\bibitem[Wren and Rousseau(1995)]{WrRo95a}
Anthony Wren and Jean-Marc Rousseau.
\newblock Bus driver scheduling--an overview.
\newblock In \emph{Computer-Aided Transit Scheduling: Proceedings of the 6th
  International Workshop on Computer-Aided Scheduling of Public Transport},
  pages 173--187, 1995.

\bibitem[Zhou et~al.(2022)Zhou, Li, and Chan]{ZhouLC22}
Houyu Zhou, Minming Li, and Hau Chan.
\newblock Strategyproof mechanisms for group-fair facility location problems.
\newblock In \emph{Proceedings of the 31st International Joint Conference on
  Artificial Intelligence (IJCAI)}, pages 613--619, 2022.

\end{thebibliography}


\appendix

\section*{Appendix}

In the appendix, we present missing proofs and additional material.

\section{Limitations of Solutions}\label{app:further}

We now strengthen \Cref{thm:effVSfair} by showing that minimum-cost solutions do not even provide approximate JR.
For \bst[0], we strengthen this result even further by showing incompatibility between approximate JP and approximate cost-minimization.
We start by formally defining our approximation objectives.

\begin{definition}
    Let $\appr \ge 1$.  
    A solution $S\subseteq V$ is said to provide \emph{$\appr$-justified representation ($\appr$-JR)} if for every set of agents $M\subseteq N$ with $|M|\ge \appr\cdot\frac {2n}b$ and every pair of stops $T\subseteq V$ there exists an agent $i\in M$ such that $c_i(T) \ge c_i(S)$.
\end{definition}

\begin{definition}
    Let $\costappr \ge 1$. A solution $S^*\subseteq V$ is said to provide a {\em $\costappr$-approximation of minimum cos}t if $c(S^*) \le \costappr\cdot c(S)$ for all solutions $S$ with $|S| = |S^*|$.
\end{definition}

We are ready to prove our inapproximability result.

\begin{theorem}\label{thm:effVSapproxJR}
    Let $\alpha \in [0,1)$ and $\appr, \costappr\ge 1$.
    Then there exists an instance of \bst{} such that every feasible solution minimizing total cost violates $\appr$-JR.
    Moreover, there exists an instance of \bst[0] such that no feasible solution simultaneously provides $\appr$-JR and $\costappr$-approximation of minimum cost.
\end{theorem}

\begin{proof}
    Let $\alpha\in [0,1)$.
    Consider the instance $\langle N,V,b, (\theta_i)_{i\in N}\rangle$, depicted in \Cref{fig:effVSapproxJR}.
    The instance is defined based on 
    two parameters $k,\ell\in \mathbb N$; the values of these parameters will be specified later. 
    Let $N_A = \{a_i\colon i\in [k]\}$, $N_D = \{d_i\colon i\in [\ell]\}$, and set $N=N_A\cup N_D$. Let $V = \{-1,0\}\cup\{ik\colon i\in [k+1]\}$, and set $b = k+2$. 
    Note that $b = |V| -1$.
    The agents' types are as follows: 
    For $i\in [k]$ we have $\theta_{a_i} = (0,ik)$ and $\theta_{d} = (-1,(k+1)k)$ for all $d\in N_D$.

    Consider the feasible solution $S^* = \{0\}\cup\{ik\colon i\in [k+1]\}$.
    Note that the sum of lengths of all agents' routes is $\ell((k+1)k+1) + \sum_{i = 1}^k ik =: K$.
    It holds that 
    \begin{equation}\label{eq:costBest}
        c(S^*) = \alpha K + \ell(1-\alpha)\text.
    \end{equation}

    Now, let $S\subseteq V$ be a feasible solution different from $S^*$.
    If $0\notin S$, then each agent in $N_A$ has to walk at least one step and therefore $c(S) \ge \alpha K + k(1-\alpha)$.
    If $(k+1)\notin S$, then each agent in $N_D$ has to walk at least a distance of $k$, and therefore $c(S) \ge \alpha K + \ell k(1-\alpha)$.
    Finally, if for some $i\in [k]$ it holds that $ik\notin S$, then $a_i$ has to walk for a distance of at least $k$ and it holds that $c(S) \ge \alpha K + k(1-\alpha)$.
    To summarize, we conclude that the cost of every feasible solution $S\neq S^*$ is at least 
    \begin{equation}\label{eq:costOther}
        c(S) \ge \alpha K + k(1-\alpha)\text.
    \end{equation}

    We are now ready to compute approximation guarantees.
    Consider target approximation ratios $\appr,\costappr\ge 1$, and
    set $\ell =\left\lceil 4\appr\right\rceil$ and $k = \left\lceil \costappr \ell \right\rceil + 1$.
    Note that this implies $k > \ell$.
    Hence, by \Cref{eq:costBest,eq:costOther}, it holds that $c(S^*) < c(S)$ for every feasible solution $S\neq S^*$, i.e., 
    $S^*$ is the unique feasible solution of minimum cost.
    Moreover, if $\alpha = 0$, it holds that $\frac{c(S)}{c(S^*)} \ge \frac k{\ell} > \frac{\left\lceil \costappr \ell \right\rceil}{\ell} \ge \costappr$ for every feasible solution $S\neq S^*$.
    Thus, in this case $S^*$ is the unique feasible solution that is a $\costappr$-approximation of minimum cost.

    We conclude the proof by showing that $S^*$ does not provide $\appr$-JR.
    To this end, consider $M = N_D$.
    We have 
    $$
    |M| = \ell = \left\lceil 4\appr\right\rceil \ge 4\appr > 2\cdot\left(\frac {k}{k+2} + \frac {\ell}{k+2}\right)\cdot\appr = \appr\cdot\frac{2n}{b}, 
    $$
    where we use the definition of $\ell$ and the fact that $k > \ell$.
    However, for $T = \{-1, (k+1)k\}$ and every $d\in N_D$ it holds that $c_{d}(S^*) = \alpha (k+1)k+1$, whereas $c_{d}(T) = \alpha ((k+1)k+1) < c_{d}(S^*)$.
    Hence, $S^*$ does not provide $\appr$-JR.
\end{proof}

\begin{figure*}
    \centering
    \resizebox{1\textwidth}{!}{
    \begin{tikzpicture}
        \draw (-.7,0) -- (10.7,0);
        \draw[->] (12.3,0) -- (16.7,0);
        \node at (11.5,0) {\dots};
        \foreach \i in {0,1,2,3,4,5,6,7,8,9,10,13,14,15,16}
        {
        \draw (\i,-.1) -- (\i,.1);
        }
        \foreach \i/\l in {0/-1,1/0,4/k,7/2k,10/3k,13/k^2,16/(k+1)k}
        {
        \node at (\i,-.4) {$\l$};
        }
        \node[align = center] at (-1,-.4) {terminals};
        \node[align = center] at (-1,.9) {desired\\routes};
        \draw (1,.3) edge[->] node[pos =.5, fill = white] {$a_1$} (4,.3);
        \draw (1,.6) edge[->] node[pos =.5, fill = white] {$a_2$} (7,.6);
        \draw (1,.9) edge[->] node[pos =.5, fill = white] (a3) {$a_3$} (10,.9);
        \draw (1,1.6) edge[->] node[pos =.5, fill = white] (ak) {$a_k$} (13,1.6);
        \draw (0,2) edge[->] node[pos =.53, fill = white] {$\ell$ agents} (16,2);
        \node at ($(a3)!.5!(ak)$) {\dots};
    \end{tikzpicture}
    }
    \caption{Incompatibility of approximate efficiency and approximate JR in \Cref{thm:effVSapproxJR}.}
    \label{fig:effVSapproxJR}
\end{figure*}

However, it is possible to combine JR with Pareto optimality, which is a weaker notion of efficiency.

\begin{definition}
    Given an instance $\ins = \langle N,V,b, (\theta_i)_{i\in N}\rangle$, 
    a solution $S$ is said to \emph{Pareto-dominate} another solution $S'$ if $c_i(S)\le c_i(S')$ for all $i\in N$ and there exists an agent $j\in N$ with $c_j(S) < c_j(S')$. A solution $S^*$ is {\em Pareto-optimal} for $\ins$ if it is not dominated by any other solution.
\end{definition}

\begin{proposition}\label{prop:pareto}
    Let $\alpha \in [0,1]$.
    Then every instance of \bst{} that admits a solution providing JR also admits a solution that is Pareto-optimal and provides JR.
\end{proposition}

\begin{proof}
    Let $\alpha \in [0,1]$, and consider an instance of \bst{} that admits a solution $S$ providing JR. Suppose that $S$ is not Pareto-optimal. Then it is Pareto-dominated by another solution $S_1$. We claim that $S_1$, too, provides JR. To see this, 
    consider a subset of agents $M\subseteq N$ with $|M| \ge \frac{2n}b$ and a pair of bus stops $T\subseteq V$.
    Since $S$ provides JR, there exists an agent $i\in M$ with $c_i(S)\le c_i(T)$.
    Hence, $c_i(S_1)\le c_i(S)\le c_i(T)$, which establishes that $S_1$ provides JR. If $S_1$ is not Pareto-optimal, there exists another solution $S_2$ that Pareto-dominates it, and our argument shows that $S_2$ provides JR as well. We can continue in this manner until we reach a Pareto-optimal solution; this will happen after a finite number of steps, as each step reduces the total cost.
\end{proof}

Proposition~\ref{prop:pareto} extends to approximate JR solutions; the proof remains the same. We note, however, that \Cref{prop:pareto} does not offer an efficient algorithm to find a Pareto-optimal solution that provides JR, as it is not clear how to compute Pareto improvements in polynomial time.

Next, we show that the solution corresponding to the terminal points used by the highest number of agents may fail JR.
Given an instance $\langle N, V, b, (\theta_i)_{i\in N}\rangle$ of {\bst} and $v\in V$, we define $\mathrm{supp}(v) := |\{i\in N\colon v\in \{\ell_i,r_i\}\}|$ as the \emph{support set} of $v$.
We say that a solution $S\subseteq V$ is \emph{support-maximizing} if for every $v\in V\setminus S$ and $w\in S$, it holds that $\mathrm{supp}(v) \le \mathrm{supp}(w)$, i.e., the solution selects bus stops with the highest support (breaking ties in some way).

\begin{proposition}\label{prop:support}
    Let $\alpha\in [0,1)$.
    Then, there exists an instance of {\bst} where no feasible support-maximizing solution provides JR.
\end{proposition}

\begin{proof}
    We define an instance $\langle N,V,b, (\theta_i)_{i\in N}\rangle$ where $N = \{a_1, \dots, a_{12}\}$, $V = [10]$, $b = 4$, and the agent have the following types: 
    $$\theta_{a_i} = 
    \begin{cases}
        (1,2) & 1\le i\le 3\\
        (3,4) & 4\le i\le 6\\
        (5,8) & 7\le i\le 8\\
        (6,9) & 9\le i\le 10\\
        (7,10) & 11\le i\le 12
    \end{cases}
    $$

Then, the unique feasible support-maximizing solution is $S = \{1,2,3,4\}$.

However, for $M = \{a_i\colon 7\le i\le 12\}$ and $T = \{7,8\}$, it holds that $|M| = \frac{2n}b$ and $c_i(T) < c_i(S)$ for all $i\in M$.
\end{proof}

\section{Proofs Missing Concerning Efficiency in Section~\ref{sec:eff}}\label{app:missing:eff}

We start with the technical lemma behind our dynamic program.

\updateLemma*
\begin{proof}
    Let $\alpha \in [0,1]$, and let $\ins = \langle N,V,b, (\theta_i)_{i\in N}\rangle$ be an instance of \bst.
    Assume that $S$, $h$, and $k$ are as in the statement of the lemma.
    We show how to compute the change in cost for agent~$i$ by adding stop~$k$ to the solution.
    For this, we make a case distinction based on the relative positions of $\ell_i$ and $r_i$ with respect to $h$ and $k$.
    The qualitative behavior is described in \Cref{fig:CostUpdates}.
    We consider the cases column-wise from top to bottom, leaving the infeasible cases for the end.
    
    \begin{figure}[hb]
        \centering
        \begin{tikzpicture}
            \pgfmathsetmacro\cellsize{2.5}
            \draw[->] (-.5,0) -- (3*\cellsize,0);
            \draw[->] (0,.5) -- (0,-3*\cellsize);
            \node at ($(3*\cellsize,0)+(-.2,.3)$) {$\ell_i$};
            \node at ($(0,-3*\cellsize)+(-.3,.2)$) {$r_i$};

            \draw[dashed] (\cellsize,-3*\cellsize) -- (\cellsize,.15);
            \draw[dashed] (2*\cellsize,-3*\cellsize) -- (2*\cellsize,.15);
            \draw[dashed] (-.15,-\cellsize) -- (3*\cellsize,-\cellsize);
            \draw[dashed] (-.15,-2*\cellsize) -- (3*\cellsize,-2*\cellsize);

            \node at (\cellsize,.45) {$h$};
            \node at (2*\cellsize,.45) {$k$};
            \node at (-.45,-2*\cellsize) {$k$};
            \node at (-.45,-\cellsize) {$h$};

            \node[align = center] at (.5*\cellsize,-.5*\cellsize) {bus has already\\ passed};
            \node[align = center] at (.5*\cellsize,-1.5*\cellsize) {saving from\\ departing};
            \node[align = center] at (.5*\cellsize,-2.5*\cellsize) {saving by\\ riding bus};
            \node at (1.5*\cellsize,-.5*\cellsize) {infeasible};
            \node at (1.5*\cellsize,-1.5*\cellsize) {single option};
            \node[align = center] at (1.5*\cellsize,-2.5*\cellsize) {saving from\\entering};
            \node at (2.5*\cellsize,-.5*\cellsize) {infeasible};
            \node at (2.5*\cellsize,-1.5*\cellsize) {infeasible};
            \node[align = center] at (2.5*\cellsize,-2.5*\cellsize) {bus has not\\ arrived yet};

            \node at (.5*\cellsize,.45) {$\ell_i < h$};
            \node at (1.5*\cellsize,.45) {$h\le \ell_i < k$};
            \node at (2.5*\cellsize,.45) {$\ell_i \ge k$};
            \node[rotate = 90] at (-.45,-.5*\cellsize) {$r_i < h$};
            \node[rotate = 90] at (-.45,-1.5*\cellsize) {$h\le r_i < k$};
            \node[rotate = 90] at (-.45,-2.5*\cellsize) {$r_i \ge k$};

        \end{tikzpicture}
        \caption{Cost updates for an agent as a function of their starting and departure destination}
        \label{fig:CostUpdates}
    \end{figure}
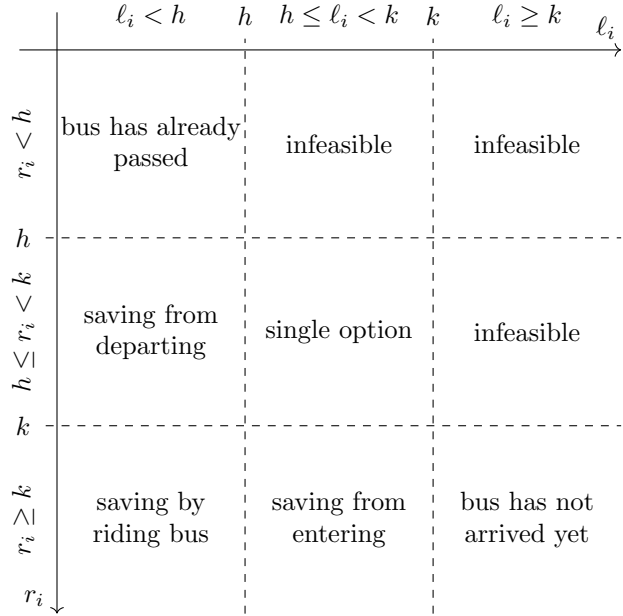

    \begin{enumerate}
        \item If $\ell_i,r_i < h$, then the optimal route for agent $i$ does not include the bus stop at $k$.
        The intuition is that the bus has already passed every possible stop that was relevant for agent $i$.
        Hence, the cost remains the same as in the solution without containing the stop $k$, i.e., $c_i(S\cup\{k\}) = c_i(S)$.
        \item If $\ell_i < h\le r_i<k$, then the route for agent $i$ is already optimal until they reach $h$.
        They have come there either by walking or by entering the bus at some stage and then taking the bus until $h$.
        In the solution without $k$, agent $i$ has to cover the distance from $h$ to $r_i$ by walking.
        Now, since the terminal point of agent $i$ is between $h$ and $k$, they may replace walking from $h$ to $r_i$ by taking the bus to $k$ and walking back to $r_i$.
        If this is faster, they decrease their cost by $(r_i-h)-(k - r_i)-\alpha (k - h)$.
        In other words, $c_i(S\cup\{k\}) = c_i(S) - \max\{0,(r_i-h)-(k - r_i)-\alpha (k - h)\}$.
        \item If $\ell_i < h < k\le r_i$, then the route for agent $i$ is again optimal until they reach $h$.
        Then, they will certainly take the bus from $h$ until $k$ and reduce their cost by $(1-\alpha)(k-h)$.
        In other words, $c_i(S\cup\{k\}) = c_i(S) - (1-\alpha)(k-h)$.
        \item If $h \le \ell_i < r_i < k$, then there is no stop on the route of agent $i$.
        Moreover, before adding $k$, the optimal route was to walk.
        After adding $k$, the unique alternative route is to walk back to $h$, enter the bus there, ride the bus to $k$, and walk back to $r_i$. Thus, the old cost is $r_i-\ell_i$ and the new cost is $\alpha(k-h)+(\ell_i-h)+(k-r_i)$. The difference between the old cost and the new cost is then
        $r_i-\ell_i-\alpha(k-h)-(\ell_i-h)-(k-r_i)$
        so 
        $c_i(S\cup\{k\}) = c_i(S) -\max\{0,2(r_i-\ell_i)-(1 + \alpha)(k-h)\}$.
        \item If $h\le \ell_i < k\le r_i$, then the agent had to walk in the optimal solution without $k$.
        Moreover, when $k$ is added, agent $i$ still has to walk from $k$ to their destination.
        However, agent $i$ now has the option to walk back to $h$ and take the bus to $k$ instead of walking from $\ell_i$ to $k$.
        If this is faster, we save a cost of $(k-\ell_i)-(\ell_i-h)-\alpha(k-h)$.
        In other words, $c_i(S\cup\{k\}) = c_i(S) - \max\{0,(k-\ell_i)-(\ell_i-h)-\alpha(k-h)\}$.
        \item If $\ell_i,r_i \ge k$, then the bus has no stops covering the route that agent $i$ would like to take.
        Hence, the optimal route for agent $i$ is to walk the whole path, both with and without stop $k$, and the cost remains the same as before, i.e., $c_i(S\cup\{k\}) = c_i(S)$.
        \item The cases $r_i < h \le \ell_i$ and $r_i < k\le \ell_i$ are impossible because we assume that $\ell_i < r_i$.
    \end{enumerate}
    All the update formulas only depend on $h$ and $k$ and they are computable in polynomial time.
\end{proof}

\begin{proposition}\label{dpclaim}
        For $\dynlast\in \{0,\dots, m\}$ and $\bg\in \{0,\dots, b\}$,
        let $\opt[\dynlast, \bg]$ be the minimum cost of {\bst} if $\dynlast$ is the rightmost open stop, and a budget of at most $\bg$ is used.
        Then, $\dyn[\dynlast, \bg] = \opt[\dynlast, \bg]$.
\end{proposition} 
    
\begin{proof}
        We prove the claim by induction over $h \in [m]$.
        By the initialization of our dynamic program, the claim is true for $h = 1$.

        Now, let $2\le h\le m$ and suppose that the claim is true for $h'\in \{1,\dots, h-1\}$.
        First, by our initialization, the claim is correct for $\bg = 0$.
        Consider $\bg\in \{1,\dots, b\}$, and
        let $\dynlast^*\in \{0,\dots, \dynlast - 1$\}.
        Then, $\opt[\dynlast, \bg]\le \opt[\dynlast^*, \bg-1] - \Delta[\dynlast^*,\dynlast]$.
        Indeed, by starting with an optimal solution for $\dynlast^*$ and $\bg-1$ and adding a stop at $\dynlast$ we obtain a valid solution for $\dynlast$ and $\bg$.
        Hence, the cost of this solution is an upper bound on the minimum cost for $\dynlast$ and $\bg$.
        Consequently,
        \begin{align*}
            \opt[\dynlast, \bg] &\le \min_{\dynlast'\in \{0,\dots,\dynlast-1\}}\opt[\dynlast', \bg-1] - \Delta(\dynlast',\dynlast)\\ &= 
            \min_{\dynlast'\in \{0,\dots,\dynlast-1\}}\dyn[\dynlast', \bg-1] - \Delta(\dynlast',\dynlast)\\
            &= \dyn[\dynlast, \bg]\text,
        \end{align*}
        where the first equality  uses the induction hypothesis.

        Now, let $S$ be a solution of minimum cost where $\dynlast$ is the rightmost open stop and a budget of at most $\bg$ is used.
        Let $\dynlast^*$ be the second stop from the right in $S$ where we set $\dynlast^* = 0$ if $S = \{\dynlast\}$.
        Then, $S^* :=S\setminus \{\dynlast\}$ is a candidate solution with budget $\bg-1$ where $\dynlast^*$ is the rightmost opened stop.
        So, $c(S^*) \ge \opt[\dynlast^*, \bg-1]$.
        Hence, 
        \begin{align*}
            \opt[\dynlast, \bg] &= c(S) = c(S^*) - \Delta(\dynlast^*,\dynlast) \\ &\ge \opt[\dynlast^*, \bg-1] - \Delta(\dynlast^*,\dynlast)\\
         &= \dyn[\dynlast^*, \bg-1] - \Delta(\dynlast^*,\dynlast)\\
         &\ge \min_{\dynlast'\in \{0,\dots,\dynlast-1\}}\dyn[\dynlast', \bg-1] - \Delta(\dynlast',\dynlast)\\& = \dyn[\dynlast, \bg]\text.\qedhere
        \end{align*}
\end{proof}

Next, we consider the bus stop problem with existing bus stops.

\ThmExistingStops*

\begin{proof}
    To solve {\bst} with existing bus stops, we modify the dynamic program developed in the proof of \Cref{thm:optimal}.
    Specifically, we adjust the update formula in \Cref{eq:DPupdate}, by first checking whether $\dynlast$ is an existing stop.
If not, we update as in \Cref{eq:DPupdate}.
Otherwise, we update as
    \begin{equation*}
    \dyn[\dynlast, \bg] 
    = \min_{\dynlast'\in \{0,\dots,\dynlast-1\}}\dyn[\dynlast', \bg]
    - \Delta(\dynlast',\dynlast)\text.
    \end{equation*}
    Since we can still update a cell in time $O(m)$, we obtain the same running time as in \Cref{thm:optimal}.
\end{proof}

We continue with our hardness result for minimizing total cost if bus stop costs are represented by binary numbers.

\BinaryHard*

\begin{proof}
    It is easy to see that our problem is in {\NP}: given a solution to {\bst}, one can efficiently check if it satisfies the budget constraint and whether its cost does not exceed the cost bound.
    
    For \NP-hardness, we provide a reduction from the \NP-complete problem {\knap} \citep{Karp72a}.
    An instance of {\knap} consists of a sequence of of $k$ weights $w = (w_1,\dots,w_k)\in \mathbb N^k$, $k$ values $v = (v_1,\dots, v_k)\in \mathbb N^k$, a capacity $w^*\in \mathbb N$, and a goal value $v^*\in \mathbb N$.
    An instance is a yes-instance if there exists a set of indices $I \subseteq [k]$ such that $\sum_{i\in I} v_i\ge v^*$ while $\sum_{i\in I} w_i\le w^*$, and a no-instance otherwise.

    Briefly, our reduction proceeds as follows:
    For each item $i$ of the {\knap} instance, we introduce a pair of bus stop locations that are exactly a distance of $v_i$ apart, and one agent who wants to travel between these locations.
    Moreover, the budget required for these stops is exactly equal to $w_i$.

    Formally, given an instance $\ins^\text{Kn} = \langle w,v,w^*,v^*\rangle$ of {\knap} 
    with $w=(w_1, \dots, w_k)$, $v=(v_1, \dots, v_k)$, we construct an instance $\ins = \langle N,V,b, (\theta_i)_{i\in N}\rangle$ of \bst{} as follows.
    Let $D = 1 + 2\max_{j\in [k]} v_j$.
    We define $N = [k]$ and $b = 2w^*$. 
    Moreover, for each $i\in [k]$ we let $\ell_i = iD$ and $r_i = iD + v_i$, 
    and create an agent $i$ with type $\theta_i = (\ell_i,r_i)$.
    We set $V = \{\ell_i, r_i\colon i\in [k]\}$, where for each $i\in [k]$ the construction costs for $\ell_i$ and $r_i$ are given by $\gamma(\ell_i) = \gamma(r_i) = w_i$.
    Finally, we set $q = \sum_{i\in [k]}v_i - (1-\alpha)v^*$.

    By the choice of the parameter $D$, the distance from $\ell_i$ to the closest stop in $V$ to its left is at least $1 + \max_{j\in [k]} v_j$.
    Similarly, the distance from $r_i$ to the closest stop in $V$ to its right is at least $1 + \max_{j\in [k]} v_j$.
    Hence, since $r_i - \ell_i = v_i$, an agent will take the bus if and only if both $\ell_i$ and $r_i$ are selected.
    We refer to this fact as Observation~$(*)$.
    
    We claim that $\ins^\text{Kn}$ is a yes-instance of {\knap} if and only if $\ins$ admits a feasible solution of total cost at most $q$.

    Assume first that $\ins^\text{Kn}$ is a yes-instance of {\knap}.
    Then there exists a subset $I\subseteq [k]$ that satisfies $\sum_{i\in I} v_i\ge v^*$, $\sum_{i\in I} w_i\le w^*$.
    Consider the set $S = \{\ell_i,r_i\colon i\in I\}$.
    Then, $\sum_{v\in S}\gamma(v) = \sum_{i\in I} 2 w_i \le 2 w^* = b$.
    Hence, $S$ is a feasible solution for $\ins$.
    Moreover, each agent $i\in I$ can take the bus for a cost of $c_i(S) = \alpha v_i$.
    Hence, $c(S) \le \sum_{i\in [k]}v_i - \sum_{i\in I}(1-\alpha)v_i \le \sum_{i\in [k]}v_i - (1-\alpha)v^*$.

    Conversely, assume that $\ins$ admits a feasible solution $S\subseteq V$ of total cost at most $q$.
    Consider $I = \{i\in [k]\colon \{\ell_i, r_i\} \subseteq S\}$.
    By Observation~$(*)$, for each agent $i\in N\setminus I$ there is no faster way to commute from $\ell_i$ to $r_i$ than by walking.
    Hence, $c(S) = \sum_{i\in [k]}v_i - \sum_{i\in I}(1-\alpha)v_i$.
    Since $c(S)\le q = \sum_{i\in [k]}v_i - (1-\alpha)v^*$, it follows that $\sum_{i\in I}v_i\ge v^*$.
    Moreover, $S$ is a feasible solution for $\ins$.
    Hence, $\sum_{i\in I} w_i = \frac 12 \sum_{i\in I} [\gamma(\ell_i)+\gamma(r_i)]\le \frac 12 \sum_{v\in S}\gamma(v) \le \frac 12 b = w^*$.
    Thus, $I$ certifies that $\ins^\text{Kn}$ is a yes-instance of {\knap}.
\end{proof}

Finally, we prove that cost-minimizing solutions only have to consider the agents' terminal points.

\PropStructure*

\begin{proof}
    Let $S$ be a minimum-cost feasible solution, and let $s\in S$ be a stop that is between two consecutive agent terminals $x$ and $y$ with $x < y$. Each agent who uses $s$ either boards the bus at $s$ or disembarks at $s$, and, when walking, they can approach $s$ from the left (in which case they pass $x$) or from the right (in which case they pass $y$). On the other hand, the bus itself travels from left to right.
    Let 
    \begin{itemize}
        \item $n_{xs}$ be the number of agents that approach $s$ from the left
        and then take the bus from $s$ (towards $y$), 
        \item $n_{sy}$ be the number of agents that take the bus to $s$ and walk towards $y$, 
        \item $n_{ys}$ be the number of agents that approach $s$ from the right and then take the bus (towards $y$), and
        \item $n_{sx}$ be the number of agents that take the bus to $s$ 
        and then walk back towards $x$.
    \end{itemize}
    Note that $n_{xs}+n_{sy}+n_{ys}+n_{sx}$ is the total number of agents who make use of $s$, and hence $s$ minimizes
    \begin{align*}
        f(z) = &n_{xs} [z-x + \alpha(y-z)] + n_{sy} [\alpha(z-x) + (y-z)] \\&+ n_{ys}(y-z)(1+\alpha) + n_{sx}(z-x)(1+\alpha)
    \end{align*}
    subject to $ x \le z \le y$.
    However, $f(z)$ is a linear function of $z$ and therefore is minimized at $z = x$ or $z=y$. Hence, $S$ can be transformed to a solution of the same cost by replacing $s$ with $x$ or $y$. By applying this argument to all stops not located at agents' terminals, we obtain the desired solution.
\end{proof}

\section{Details of Experimental Analysis}

In this section, we provide further details about our experiments.

\subsection{Verification of Fair Solutions}\label{app:verification}

A crucial step of our experimental analysis to test whether the outcomes computed by our algorithm are in the core.
Our key idea to perform this computational task is that whenever we want to contest the fairness of a solution by proposing a better set of stops, this is only improving for agents that find a pair of stops within this set that leads to a lower cost.
Formally, given a solution $S$ and an agent $i\in N$, we define the set $\betterstops{S}$ as the set of pairs of stops that would lead to a cost lower than the cost of of $S$, i.e., $\betterstops{S} := \{T\subseteq V\colon |T| = 2, c_i(T) < c_i(S)\}$.
Clearly, we can express both of our fairness concepts in terms of these sets.
The proof of the next proposition follows immediately from \Cref{def:PF,def:core}.

\begin{proposition}\label{prop:verificationFair}
    Consider a solution $S\subseteq V$. Then,
    \begin{enumerate}
        \item $S$ provides JR if and only if there does not exist a set $T\subseteq V$ with $|T| = 2$ such that $|\{i\in N\colon T\in \betterstops{S}\}|\ge \frac {2n}b$.
        \item $S$ is in the core if and only if there exists no set $T\subseteq V$ with $T\neq \varnothing$ such that $|\{i\in N\colon \exists\, T'\in \betterstops{S} \text{ with } T'\subseteq T\}|\ge \frac {|T|n}b$.
    \end{enumerate}
\end{proposition}

\begin{proof}
    The first statement is a reformulation of \Cref{def:PF}. 
    For the second statement, we observe that $c_i(T) < c_i(S)$ if and only if there exists a $T'\subseteq T$ with $|T'| = 2$ (namely the two stops that lead to the cost of $c_i(T)$) such that $c_i(T') < c_i(S)$.
\end{proof}

The first part of \Cref{prop:verificationFair} immediately implies that 
one can test whether a solution $S$ provides JR in polynomial time:
One can simply compute the set $\betterstops{S}$ for each agent, and then check for every pair $T\subseteq V$ with $|T| = 2$ whether the condition of \Cref{prop:verificationFair} is satisfied.

\begin{corollary}
    Testing whether a given solution provides JR can be performed in polynomial time.
\end{corollary}

The same approach does not necessarily lead to a polynomial-time algorithm for the core because this would require to check an exponential number of sets.
Instead, for our implementation, we perform this task by the following integer program.

\begin{align*}
    &(\textsc{CoreTesting})&
     \\
    \text{max} \quad \sum_{i\in N} x_i &
     \\
    \text{s.t. } \quad  x_i &\le \sum_{T'\in \betterstops{S}} y_{T'} & \forall\, i\in N
     \\
     \quad y_{T'} &\le y_s & \forall\, T'\subseteq V, |T'| = 2, s\in T'
     \\
     \quad \sum_{i\in N} x_i &\ge \frac nb\sum_{s\in V} y_s 
     \\
    x_i & \in \{0,1\} & \forall\, i\in N
    \\
    y_s & \in \{0,1\} & \forall\, s\in V
    \\
    y_{T'} & \in \{0,1\} & \forall\, T'\subseteq V, |T'| = 2
    \\
\end{align*}

The correctness of this integer program follows from our next proposition.

\begin{proposition}
    A solution $S$ is in the core if and only if its corresponding integer program (\textsc{CoreTesting}) has an optimal value of $0$.
\end{proposition}

\begin{proof}
    The integer program (\textsc{CoreTesting}) contains three types of binary variables.
    The variable $x_i$ for $i\in N$ indicates whether agent $i$ belongs to a deviating coalition of agents. 
    The variable $y_s$ for $s\in V$ indicates whether the deviating coalition pays for opening a stop at position $s$.
    Finally, for every $T'\subseteq V$ with $|T'| = 2$, the variable $y_{T'}$ indicates whether the deviating coalition pays for opening both stops in $T'$.
    
    The constraint $x_i \le \sum_{T'\in \betterstops{S}} y_{T'}$ ensures that an agent is only part of the deviating coalition if the opened stops actually lead to an improvement.
    The constraint $y_T'\le y_s$ ensures that a pair of stops is only treated as open if both elements contained in this set are indicated as open.
    Finally, the constraint $\sum_{i\in N} x_i \ge \frac nb\sum_{s\in V} y_s$ ensures that the deviating coalition is large enough to pay for all stops they want to open.

    Hence, a feasible solution of (\textsc{CoreTesting}) gives rise to a set of agents $M = \{i\in N\colon x_i = 1\}$ and a set of stops $T = \{s\in V\colon y_s = 1\}$ such that $M \subseteq \{i\in N\colon \exists\, T'\in \betterstops{S} \text{ with } T'\subseteq T\}$ (by the first two constraints).
    If the optimal solution has value different from $0$, then $T \neq \varnothing$ (otherwise, the first constraint cannot be satisfied) and $|\{i\in N\colon \exists\, T'\in \betterstops{S} \text{ with } T'\subseteq T\}|\ge \frac {|T|n}b$ (by the third constraint).
    By \Cref{prop:verificationFair}, $S$ is not in the core.

    Conversely, if $S$ is not in the core, we use \Cref{prop:verificationFair} to find $T \subseteq V$ with $T\neq \varnothing$ and $|\{i\in N\colon \exists\, T'\in \betterstops{S} \text{ with } T'\subseteq T\}|\ge \frac {|T|n}b$.
    We set $x_i = 1$ if and only there exists $T'\in \betterstops{S}$ with $T'\subseteq T$, $y_s = 1$ if and only if $s\in T$, and $y_{T'} = 1$ if and only if $T'\subseteq T$.
    This leads to a feasible solution of objective value greater than~$0$.
    Hence, the optimal solution has an objective value greater than~$0$.
\end{proof}

\subsection{Detailed Analysis of Frequency of Fairness Violations}\label{app:DetailedSimu}
\begin{figure*}[tb]
    \centering
    \includegraphics[height = 6.7cm]{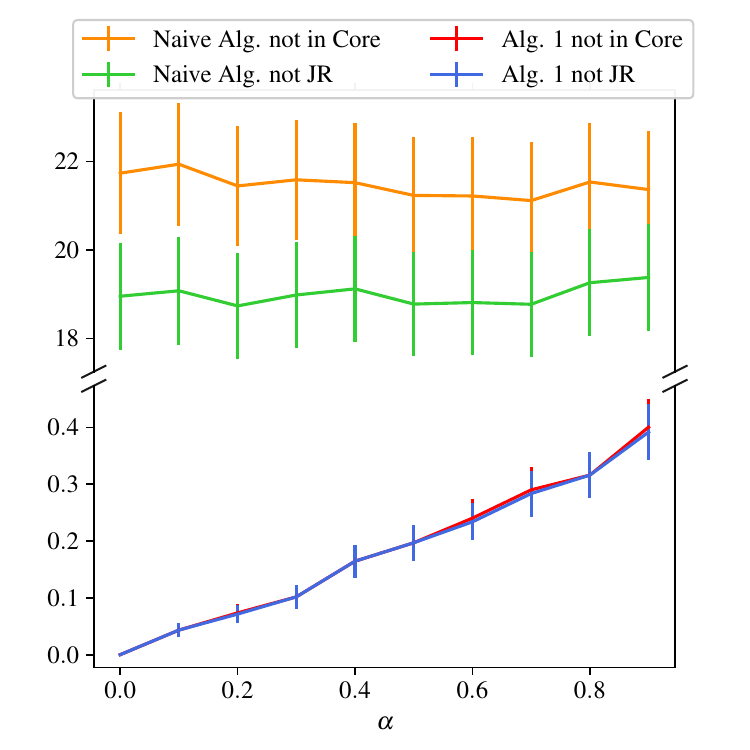}\quad
    \includegraphics[height = 6.7cm]{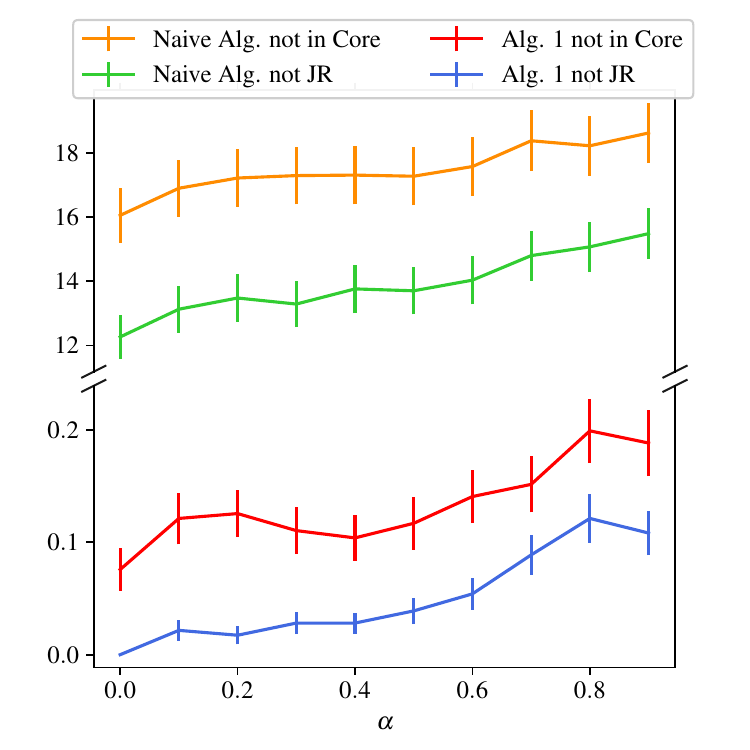}\\
    \includegraphics[height = 6.7cm]{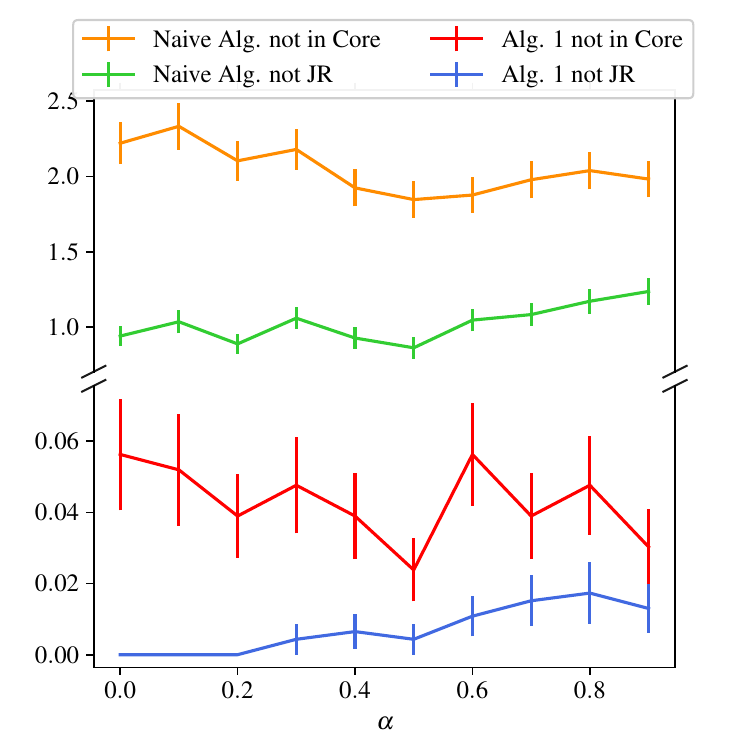}\quad
    \includegraphics[height = 6.7cm]{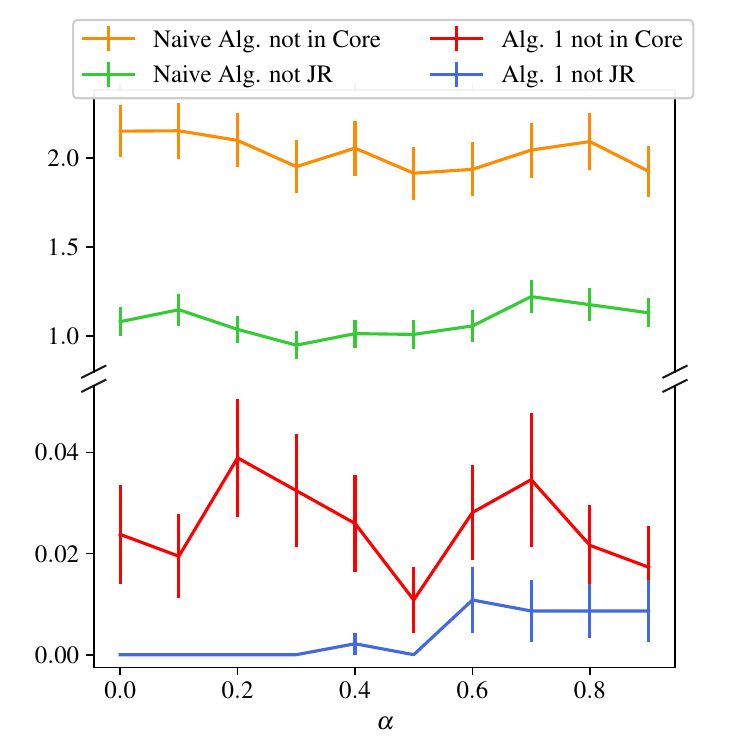}\\
    \caption{Aggregated frequency of fairness violations of the solutions computed by \Cref{alg:PF0bst} and our benchmark algorithm. 
    We aggregate the frequency for a fixed number of $5$, $10$, $15$, and $20$ agents (from top left to bottom right). 
    The $x$-axis shows our range for the cost parameter $\alpha$ and the $y$-axis shows the percentage of the instances in which the desired property is not satisfied along with the standard error.
    }
    \label{fig:FairnessViolationFixedN}
\end{figure*}

We now present further results of our experiments.
In \Cref{sec:experiments}, we have presented an aggregated view of the frequency with which \Cref{alg:PF0bst} and our naive benchmark algorithm violate fairness concepts.
\Cref{fig:FairnessViolationFixedN} complements \Cref{fig:frequency} and provides snapshots of this analysis for fixed numbers of agents.
For both algorithms, violations of fairness mostly happen for a small number of agents.
When the number of agents increases, the violations become significantly rarer, even though \Cref{alg:PF0bst} still outperforms our benchmark algorithm.
For a very small number of agents, \Cref{alg:PF0bst} produces solutions not in the core almost only when they already violate JR.
An exact comparison of the relative performance is provided in \Cref{tab:comparison:algos}, where we display the ratio of the fairness violations of both algorithms.

\begin{table}
	\caption{Relative performance of \Cref{alg:PF0bst} compared to the naive algorithm.
 Each entry of the table is the ratio of the number of fairness violations of the naive algorithm and the number of fairness violations of \Cref{alg:PF0bst}. Since \Cref{alg:PF0bst} always produces JR outcomes, the first entry is unbounded.}
	\label{tab:comparison:algos}
\centering
	\renewcommand{\arraystretch}{1.2}
\begin{tabular}{l cccccccccc}
\toprule
Cost parameter $\alpha$ & 0.0 & 0.1 & 0.2 & 0.3 & 0.4 & 0.5 & 0.6 & 0.7 & 0.8 & 0.9\\
\midrule
Violation of JR & $\infty$ & 644 & 489 & 332 & 221 & 178 & 141 & 117 & 88 & 87 \\ [0.3ex]
Violation of core & 230 & 192 & 193 & 170 & 159 & 144 & 119 & 109 & 87 & 83 \\ [0.3ex]
\bottomrule
\end{tabular}
\end{table}

\begin{figure*}[h]
    \centering
    \resizebox{1\textwidth}{!}{
    \includegraphics[trim={0.7cm 0 0.1cm 0},clip,height = 4cm]{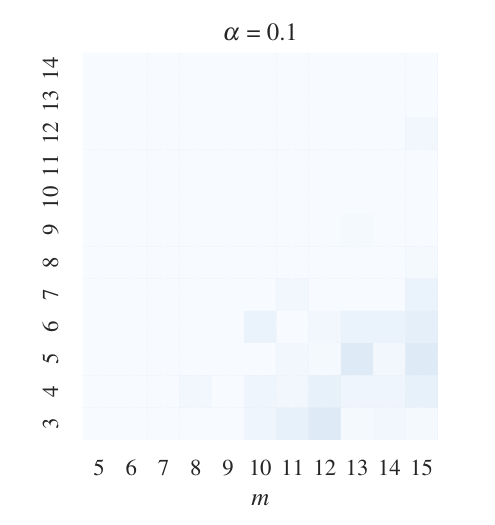}\hspace*{-.5cm}
    \includegraphics[trim={0.7cm 0 0.1cm 0},clip,height = 4cm]{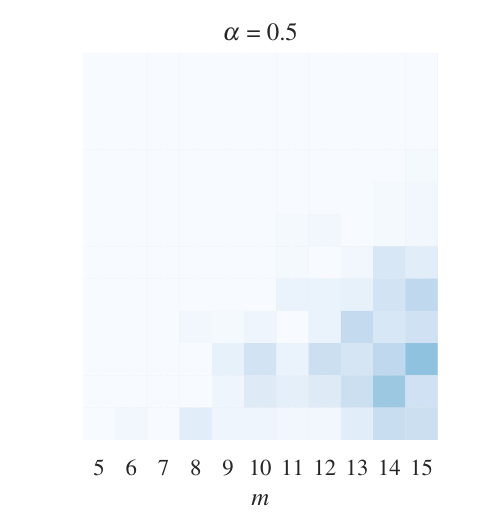}\hspace*{-.5cm}
    \includegraphics[trim={0.7cm 0 0.1cm 0},clip,height = 4cm]{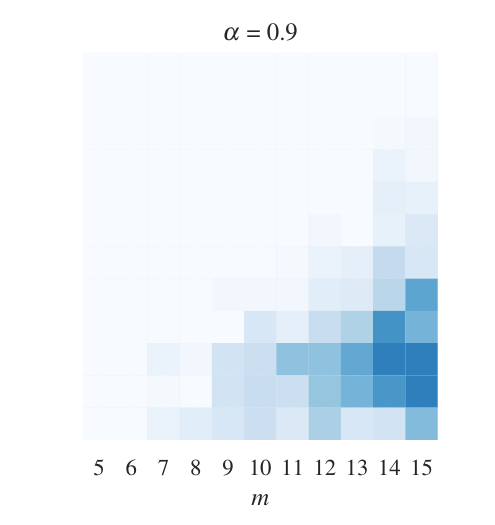}
    \hspace*{-.5cm}
    \includegraphics[trim={6.4cm 0 0.1cm 0},clip,height = 4cm]{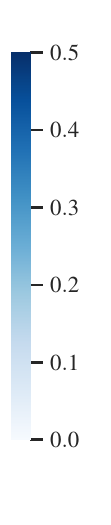}}
    \caption{Heat maps showing the frequency of JR violations of the solutions computed by \Cref{alg:PF0bst} for pairs of $m$ and $b$ in instances of \bst[0.1], \bst[0.5], and \bst[0.9]. 
    Each cell is averaged over all values of $n$.}
    \label{fig:otheralpha:PF}
\end{figure*}

\begin{figure*}[h]
    \centering
    \resizebox{1\textwidth}{!}{
    \includegraphics[trim={0 0 0.1cm 0},clip,height = 4cm]{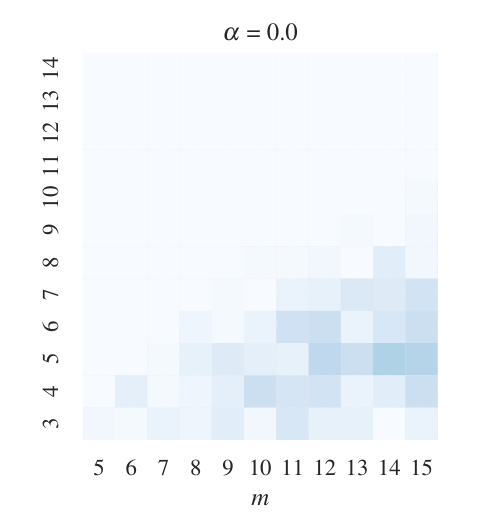}\hspace*{-.5cm}
    \includegraphics[trim={0.7cm 0 0.1cm 0},clip,height = 4cm]{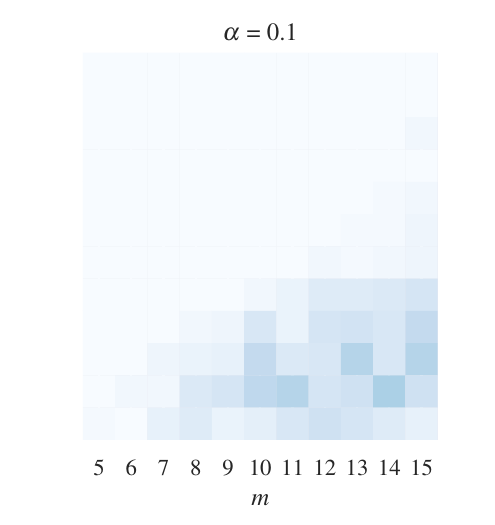}\hspace*{-.5cm}
    \includegraphics[trim={0.7cm 0 0.1cm 0},clip,height = 4cm]{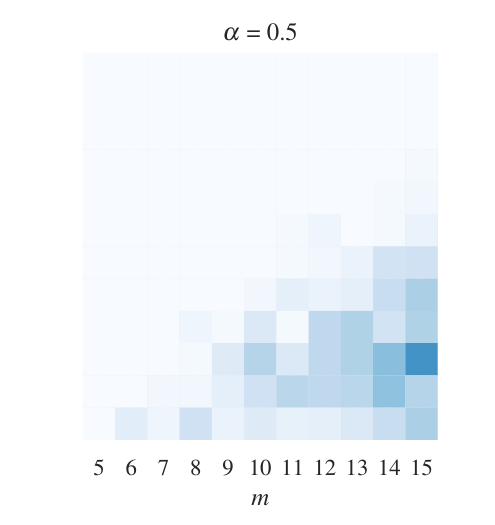}\hspace*{-.5cm}
    \includegraphics[trim={0.7cm 0 0.1cm 0},clip,height = 4cm]{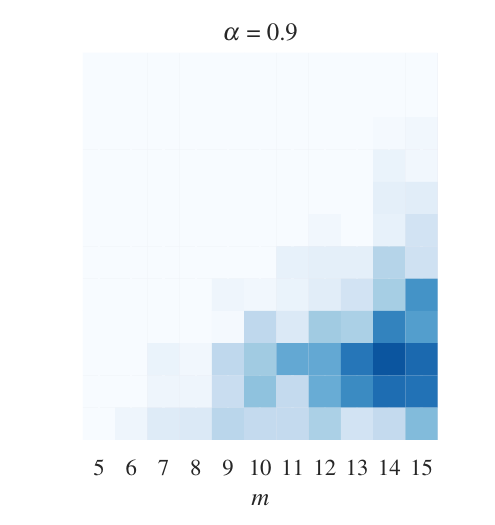}
    \hspace*{-.5cm}
    \includegraphics[trim={6.4cm 0 0.1cm 0},clip,height = 4cm]{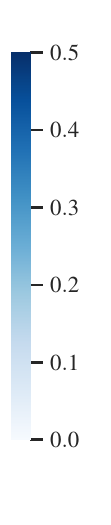}}
    \caption{Heat maps showing the frequency of core violations of the solutions computed by \Cref{alg:PF0bst} for pairs of $m$ and $b$ in instances of \bst[0], \bst[0.1], \bst[0.5], and \bst[0.9] (columns) and for JR and the core (rows). 
    Each cell is averaged over all values of $n$.}
    \label{fig:otheralpha:core}
\end{figure*}

Finally, we provide heat maps similar to the one in \Cref{fig:avg_heatmap} for different values of the cost parameter $\alpha$.
\Cref{fig:otheralpha:PF} and \Cref{fig:otheralpha:core} display the frequency of JR and core violations, respectively.
For both fairness notions, we display the cases for $\alpha = 0.1$, $\alpha = 0.5$, and $\alpha = 0.9$.
For the core, the first picture is a version of \Cref{fig:avg_heatmap} where we adjusted the color scale to facilitate the comparison with other values of $\alpha$.
Since \Cref{alg:PF0bst} always computes solutions that provide JR for \bst[0], we omit a fourth picture for JR.

Both picture series replicate the impression that we obtained from our simulations thus far:
Fairness violations seem to happen more frequently when we increase the cost parameter $\alpha$ (but are still extremely low) and mostly happen for larger $m$ and smaller $b$.

\subsection{Second Benchmark}\label{app:benchmark2}

As a second benchmark, we consider another naive algorithm based on maximizing support, an idea that we had already considered in \Cref{prop:support}.
Maximizing support can be performed by the greedy algorithm of iteratively selecting the potential stop that has the maximum support among the remaining ones. 
The performance of this algorithm compared to \Cref{alg:PF0bst} is shown in \Cref{fig:greedy}. 
It does not provide JR outcomes for more than $70\%$ of the instances and performs much worse than the first benchmark analzed in \Cref{sec:experiments}.

\begin{figure*}[h]
\centering
\begin{subfigure}[t]{0.7\textwidth}
    \centering
    \includegraphics[height=7cm]{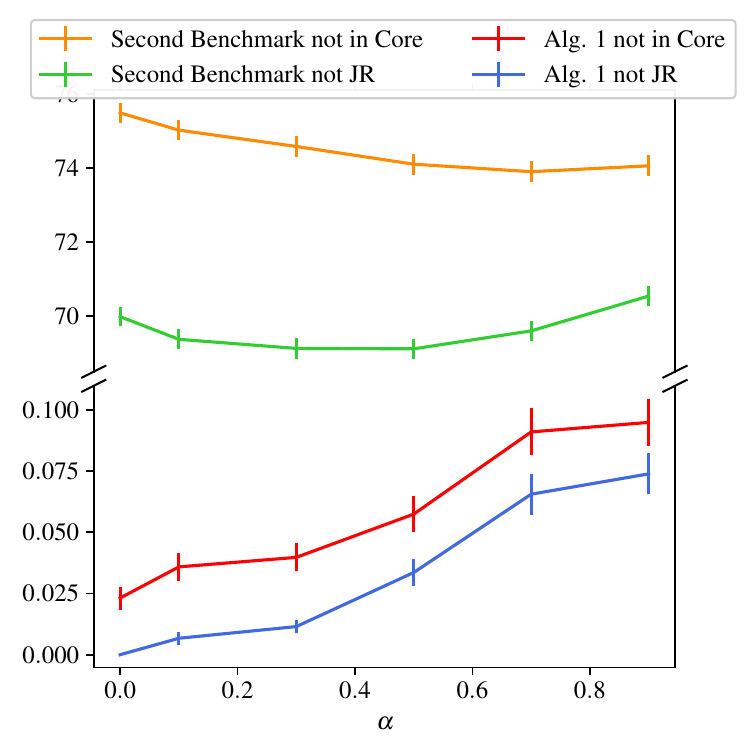}
    \caption{Aggregated frequency of fairness violations of the solutions computed by our proposed and the greedy algorithm along with the standard error.}
    \label{fig:greedyFrequency}
\end{subfigure}
~
\vspace{0.5cm}

\begin{subfigure}[b]{\textwidth}
    \centering
    \includegraphics[height = 6cm]{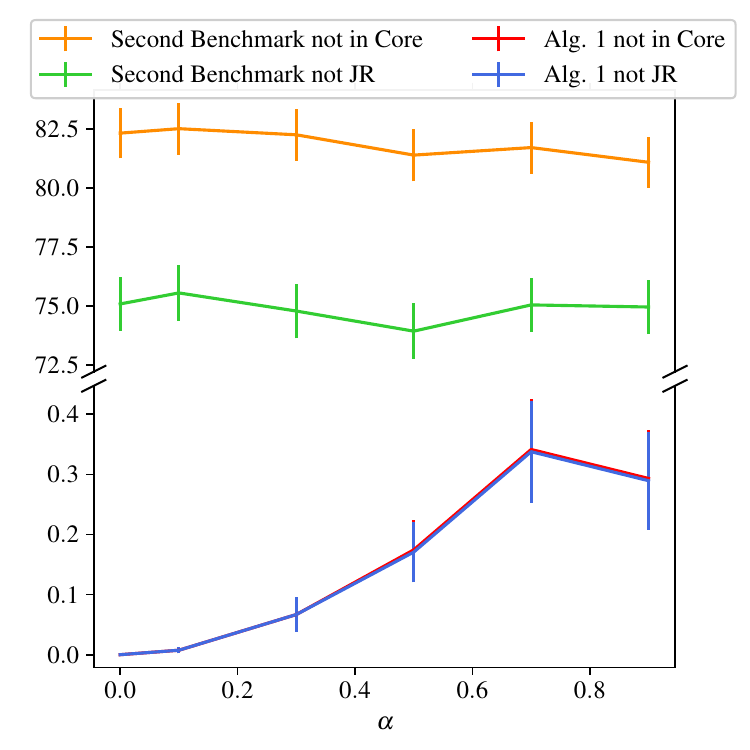}\quad
    \includegraphics[height = 6cm]{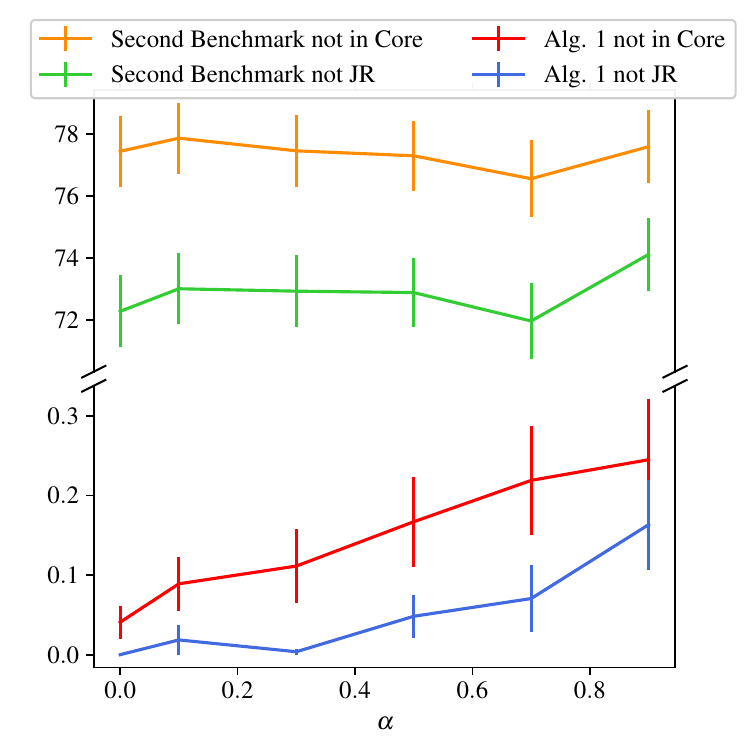}\\
    \includegraphics[height = 6cm]{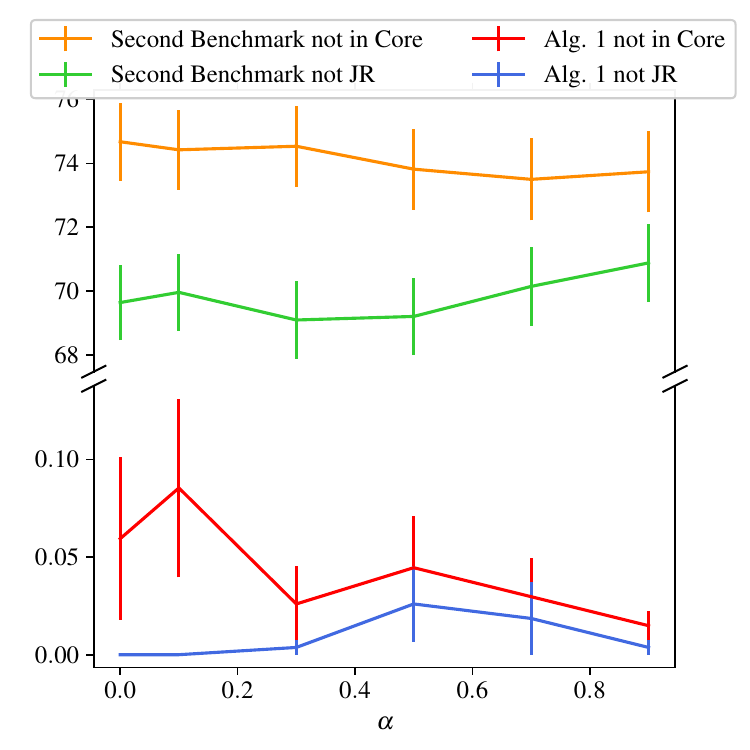}\quad
    \includegraphics[height = 6cm]{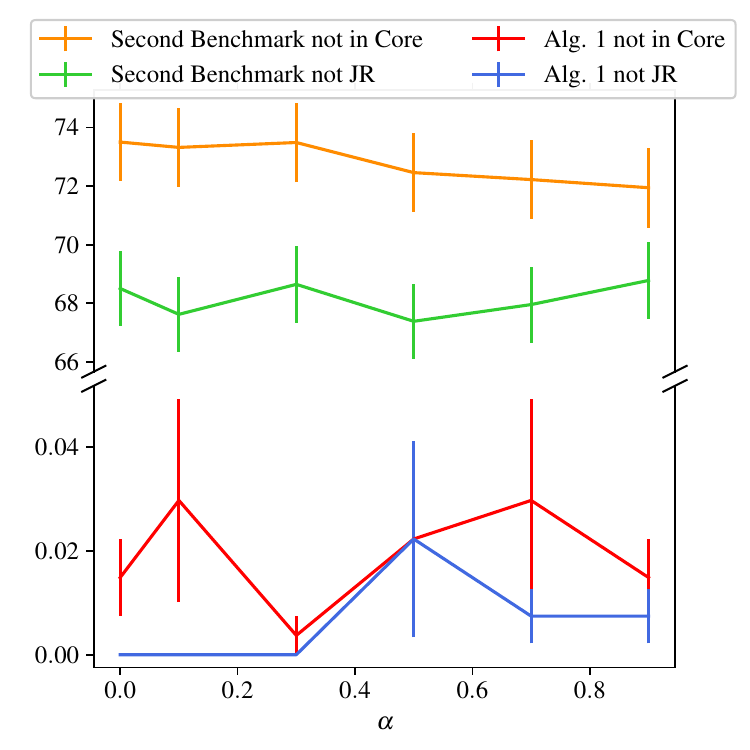}\\
    \label{fig:greedyFixed}
    \caption{Aggregated frequency of fairness violations of the solutions computed by our proposed algorithm and the greedy algorithm. 
    We aggregate the frequency for a fixed number of $5$, $10$, $15$, and $20$ agents (from top left to bottom right).}
\end{subfigure}
\caption{Experimental results comparing the greedy algorithm and \Cref{alg:PF0bst}.}
\label{fig:greedy}
\end{figure*}

\section{Approximately Stable Committee Selection}
\label{sec:stable_committee}

In this section we give an overview of the definitions and results by \citet{JiangMW20}. 
They consider the committee selection problem with a set of $n$ agents $N$ and a set of $m$ candidates $C$. 
Each candidate $c \in C$ has a weight of $s_c$. 
The goal is to find a committee of total weight at most $K$. 
In this setting, each agent $i\in N$ specifies a weak order $\succsim_i$ over all possible committees. 
Then, $\succ_i$ denotes strict preferences.
They assume that these orders respect monotonicity, which means that for committees $S_1 \subseteq S_2$ and any $i\in N$, we have $S_2 \succsim_i S_1$.
The interpretation of monotonicity is that additional candidates cannot harm.

They define stable committees and approximately stable committees as follows which is similar to our definitions.

\begin{definition}[\citealp{JiangMW20}]
Given two committees $S_1, S_2 \subseteq C$ the \emph{pairwise score} $V(S_1, S_2)$ of $S_2$ over $S_1$ is the number of voters who strictly prefer $S_2$ to $S_1$, i.e., $V(S_1, S_2) := |\{i \in N\colon S_2 \succ_i S_1\}|$.
\end{definition}

\begin{definition}[\citealp{JiangMW20}]
Given a committee $S \subseteq C$ of weight at most $K$, a committee $S' \subseteq C$ of weight $K'$ \emph{blocks} $S$ if and only if $V(S, S') \ge \frac{K'}{K}\cdot n$. 
A committee $S$ is said to be \emph{stable} (or lies in the \emph{core}) if no committee blocks it.
\end{definition}

\begin{definition}[\citealp{JiangMW20}]
Given a parameter $c \ge 1$ and a committee $S \subseteq C$ of weight at most $K$, we say that a committee $S' \subseteq C$ of weight $K'$ \emph{$c$-blocks} $S$ if and only if $V(S, S') \ge c\cdot \frac{K'}{K}\cdot n$. 
A committee $S$ is said to be \emph{$c$-approximately stable} if no committee $c$-blocks it.
\end{definition}

They prove that a $32$-approximately stable committee always exists.

\begin{theorem}[\citet{JiangMW20}]
    For any monotone preference structure with $n$ agents and $m$ candidates, arbitrary weights and the cost-threshold $K$, a 32- approximately stable committee of weight at most $K$ always exists.
\end{theorem}

This result can be improved to a $16$-approximation if the candidates are unweighted which is the case in our problem.

In \bst{}, the cost functions of the agents induce a monotone ordering over the possible subsets of the potential stops because additional bus stops can only lower the cost.
Moreover, the problem is to select $b$ stops out of $V$, and hence we have an instance of the stable committee problem for which a $16$-approximately stable solution exists.

\end{document}